\documentclass[submission, Phys]{SciPost}

\usepackage{amsfonts}
\usepackage{amsmath}
\usepackage{amsthm}
\usepackage{amssymb}
\usepackage{stmaryrd}
\usepackage{amscd}
\usepackage{eucal}
\usepackage{dsfont}
\usepackage{bm}
\usepackage{graphicx}

\numberwithin{equation}{section}

\usepackage{hyperref}
\hypersetup{
     colorlinks = true,
     linkcolor = blue,
     anchorcolor = blue,
     citecolor = blue,
     filecolor = blue,
     urlcolor = blue
     }

\usepackage{booktabs}

\input{macro.lib}

\newcommand{\nex}{\ell}

\begin{document}

\begin{center}{\Large \textbf{Spin conductivity of the XXZ chain\\*[0.1cm]
in the antiferromagnetic massive regime}}\end{center}

\begin{center}
Frank G\"ohmann,\textsuperscript{1}
Karol K. Kozlowski,\textsuperscript{2}
Jesko Sirker,\textsuperscript{3}
Junji Suzuki\textsuperscript{4}
\end{center}

\begin{center}
{\bf 1} Fakult\"at f\"ur Mathematik und Naturwissenschaften, Bergische Universit\"at Wuppertal, 42097 Wuppertal, Germany
\\
{\bf 2} Univ Lyon, ENS de Lyon, Univ Claude Bernard, CNRS, Laboratoire de Physique, F-69342 Lyon, France
\\
{\bf 3} Department of Physics and Astronomy and Manitoba Quantum Institute, University of Manitoba, Winnipeg R3T 2N2, Canada
\\
{\bf 4} Department of Physics, Faculty of Science, Shizuoka University, Ohya 836, Suruga, Shizuoka, Japan
\end{center}

\begin{center}
\today
\end{center}

%%%%%%%%%%%%%%%%%%%%%%

\section*{Abstract}
We present a series representation for the dynamical two-point 
function of the local spin current for the XXZ chain in the
antiferromagnetic massive regime at zero temperature. From this
series we can compute the correlation function with very high
accuracy up to very long times and large distances. Each term in
the series corresponds to the contribution of all scattering
states of an even number of excitations. These excitations can
be interpreted in terms of an equal number of particles and holes.
The lowest term in the series comprises all scattering states
of one hole and one particle. This term determines the long-time
large-distance asymptotic behaviour which can be obtained
explicitly from a saddle-point analysis. The space-time
Fourier transform of the two-point function of currents
at zero momentum gives the optical spin conductivity of
the model. We obtain highly accurate numerical estimates for
this quantity by numerically Fourier transforming our data.
For the one-particle, one-hole contribution, equivalently
interpreted as a two-spinon contribution, we obtain an exact
and explicit expression in terms of known special functions.
For large enough anisotropy, the two-spinon contribution carries
most of the spectral weight, as can be seen by calculating
the f-sum rule.

\vspace{10pt}
\noindent\rule{\textwidth}{1pt}
\tableofcontents\thispagestyle{fancy}
\noindent\rule{\textwidth}{1pt}
\vspace{10pt}

%%%%%%%%%%%%%%%%%%%%%%
%%%%%%%%%%%%%%%%%%%%%%
\section{Introduction}
\label{Intro}
%%%%%%%%%%%%%%%%%%%%%%
%%%%%%%%%%%%%%%%%%%%%%
Transport phenomena in spatially one-dimensional quantum
systems are an active area of research, both theoretically
and experimentally \cite{BHKPSZ21,NDMP22,SPA11,Sirker20,
Hlubeketal10,FKECSHBSGGBK13, FSEHCBG13,HFSZKDBG14,Jepsenetal20,
Weietal21pp}.
%\textcolor{blue}{\sout{
%The theoretical emphasis in recent years was
%on a systematic justification of phenomenological approaches
%and on numerical work
%in a setting that often attempted to
%go beyond the framework of linear response}}.
Some of the more prominent one-dimensional models are
integrable and therefore amenable to an exact treatment.
In linear response, their transport properties
are determined by the dynamical correlation functions of
current densities. In this work we shall focus on the
XXZ spin-1/2 chain with Hamiltonian
\begin{equation} \label{hxxz}
     H = J \sum_{j = 1}^L \Bigl\{ \s_{j-1}^x \s_j^x + \s_{j-1}^y \s_j^y
               + \D \bigl( \s_{j-1}^z \s_j^z - 1 \bigr) \Bigr\}
	       - \frac{h}{2} \sum_{j=1}^L \s_j^z \epc
\end{equation}
where the $\s^\a \in \End {\mathbb C}^2$, $\a = x, y, z$, are
Pauli matrices. The three real parameters of the Hamiltonian are
the anisotropy $\D$, the exchange interaction $J > 0$, and the
strength $h > 0$ of an external longitudinal magnetic field.

The basic quantities that can be transported in the XXZ chain
are heat and magnetization. The total heat current of the
XXZ chain is a conserved quantity. This implies that the
corresponding thermal conductivity is purely ballistic and
is determined entirely by a thermal Drude weight that can be
calculated exactly at any temperature for any value of $\D$
and $h$ \cite{KlSa02,SaKl03,Zotos17}. Technically, the thermal
Drude weight can be inferred from the spectral properties of
a properly defined quantum transfer matrix \cite{Suzuki85,%
SAW90,Kluemper93,KlSa02}.

For the current of the magnetization the situation
is different. The total spin current is not conserved, except in
the free Fermion case $\D = 0$. Still, it may have a ballistic
contribution. In that case the corresponding conductivity
consists of a singular `dc part' quantified by a Drude weight
and a regular $\om$-dependent `ac part', where $\om$ is the
frequency. There is a vast body of literature on the numerical
calculation of both, the singular and the regular contribution,
at finite temperature $T > 0$ (for an overview see the recent
review article \cite{BHKPSZ21}). On the analytical side, the
$T=0$ Drude weight in the critical regime of the XXZ chain is
known \cite{ShSu90}. Results for the Drude weight at
finite $T$ and on the leading asymptotic behaviour of the regular
part for $\om \rightarrow 0$ are also available. In particular,
exact lower bounds, based on the Mazur inequality \cite{Mazur69},
were established in both cases \cite{ZNP97,Prosen11,%
PereiraPasquier,ProsenIlievski,SPA11, MKP17,IDMP18}. For the
Drude weight in the regime $-1<\Delta<1$ at magnetic field $h=0$
it has been argued that the Mazur bound obtained by taking all
known families of conserved charges into account is tight. This
bound, furthermore, does agree with earlier results for the
Drude weight based on an extension of the thermodynamic Bethe
ansatz \cite{FujimotoKawakami,Zotos99,BFKS05}. In the latter
case, the input invoked from the Bethe Ansatz solution of the
XXZ chain \cite{Orbach58,BVV83} enters in the form of
the string hypothesis \cite{Gaudin71,TaSu72}, whose applicability
is not established beyond the calculation of thermodynamic
quantities, where its use is equivalent to the use of fusion hierarchies
\cite{KSS98,TSK01}. For the low-energy excitations of the XXZ chain
over the degenerate ground state in the antiferromagnetic
massive regime of the phase diagram, considered in this
work, it is known to give an incorrect description \cite{BVV83,%
ViWo84,DGKS15a}.

It is important to stress that all these works deal with the
conductivity of the XXZ chain in the limit $\omega\to 0$; we are
not aware of any exact result for genuine finite frequencies in the
literature. In any case, our work presented below is independent of
and rather orthogonal to the previous results in that, instead of
$T > 0$, $\om = 0$, we consider $T = 0$, $\om > 0$ in the
framework of an exact calculation of a dynamical correlation
function that does not involve any kind of string hypothesis.
%\textcolor{blue}{\sout{
%The calculation of the
%spin conductivity requires an honest calculation of the dynamical
%correlation function of two spin-current density operators.}}

So far, the most successful attempts to exactly calculate
dynamical correlation functions of the XXZ chain were based
on different types of form factor series expansions. Besides
the series involving form factors of the Hamiltonian
\cite{JiMi95,BCK96,BKM98,CaHa06,CKSW12,CaMa05,CMP08,KKMST12,%
Kozlowski18,Kozlowski21}, there is a different type of series that
utilizes form factors of the quantum transfer matrix of the
model \cite{GKKKS17}. The latter type has been dubbed the
thermal form factor expansion. It was designed to deal with
the canonical finite temperature case, but, in principle, can
be extended to include certain generalized Gibbs ensembles
\cite{FaEs13}. The thermal form factor expansion has not been
much explored so far, but it seems to have certain advantages
over the more conventional expansions employing the form factors
in a Hamiltonian basis. In \cite{GKS20a,GKS20b} it was observed
that it can have rather nice asymptotic properties as compared to
the conventional representation \cite{CIKT92,IIKS93b}
and can be interpreted as a resummation of the latter. For the
XXZ chain in the massive antiferromagnetic regime, we observed
a remarkable simplification of the Bethe root patterns occurring
in the low-temperature limit \cite{DGKS15b}. In this limit all
string excitations disappear and the whole spectrum of the
quantum transfer matrix can be interpreted in terms of
particle-hole excitations. This made it possible to derive
a series representation for the longitudinal correlation functions
of the XXZ chain in the massive antiferromagnetic regime, in
which the $n$th term comprises all scattering states of $n$
particles and $n$ holes in a $2n$-fold integral with an integrand
that is explicitly expressed in terms of known special functions
\cite{BGKS21a,BGKSS21}. This has to be contrasted with older
representations in which the integrand for general $n$ is
itself a sum over multiple integrals \cite{JiMi95}, a multiple
residue \cite{DGKS15a} or, in the best case, a product of explicit
functions and Fredholm determinants \cite{DGKS16b}. It is the simple
explicit form in \cite{BGKS21a,BGKSS21} which makes the higher-$n$
terms in the form factor series efficiently computable. The
observed fast convergence of the series, on the other hand,
is not a feature of the thermal form factor approach, but can
be attributed to the massive nature of the excitations.

In Ref.~\cite{GKKKS17}, thermal form factor series were introduced
with the example of operators of `length one', where the length
of an operator is defined as the number of lattice sites on which
it acts non-trivially. The local Pauli matrices $\s_j^\a$ are
examples of such operators. The local spin current
\begin{equation}
     {\cal J}_j = - 2 \i J (\s_{j-1}^- \s_j^+ - \s_{j-1}^+ \s_j^-)
\end{equation}
has length two. The results of Ref.~\cite{GKKKS17} can be rather naturally
extended to operators of arbitrary length. Although the general
formula is easy to guess, its proof is slightly technical. It will be
presented in a forthcoming publication. For a subclass of these operators
(the `spin-zero' operators) we can then introduce certain `properly
normalized form factors' which can be related to the theory of
factorizing correlation functions and the Fermionic basis approach
\cite{BGKS06,BGKS07,BJMST08a,JMS08,BoGo09}. This allows to
obtain series representations which, in the antiferromagnetic
massive regime, are as explicit as the series for the longitudinal
correlation functions obtained in \cite{BGKS21a,BGKSS21}. Prior
to working out the general theory, we shall present here our results
for the current-current correlation functions $\<{\cal J}_1 (t)
{\cal J}_{m+1}\>$ that may be of particular interest to the
physics community. For this special case the result may be
obtained with moderate effort by combining \cite{BGKS07,BoGo09} 
with \cite{GKKKS17,BGKS21a,BGKSS21}. One has to start with a
two-site generalized density matrix involving a twist or
`disorder parameter' $\a$ and then use the $R$-matrix symmetry,
that imposes a set of quadratic relations on the two-site
generalized density matrix, together with the reduction relation
for the latter. The general case is harder and requires the
techniques introduced in \cite{BJMST08a,JMS08}.

%%%%%%%%%%%%%%%%%%%%%%
%%%%%%%%%%%%%%%%%%%%%%
\section{Two-point function of currents}
\label{Sec2}
%%%%%%%%%%%%%%%%%%%%%%
%%%%%%%%%%%%%%%%%%%%%%
\subsection{Form factor series}
\label{FF}
The antiferromagnetic massive regime of the ground state phase
diagram of Hamiltonian~(\ref{hxxz}) is defined by the inequalities
$\D = \ch(\g) > 1$ and $|h| < h_\ell = 4 J \sh(\g) \dh_4^2 (0|q)$
for $\g > 0$. Here we have set $q = \re^{- \g}$, and $\dh_4$ denotes
a Jacobi theta function. For the Jacobi theta functions, that will
be frequently needed below, we shall follow the conventions of
Ref.~\cite{WhWa63ch21}, see \eqref{definition theta 4}-\eqref{definition theta 1 2 3}
for a reminder. Other special functions that occur in the
definition of the form factor amplitudes belong to the
families of $q$-gamma and $q$-hypergeometric (or basic
hypergeometric) functions. We list their definitions and some
of their properties in Appendix~\ref{app:specialfun}.

In order to be able to present our series representation for the
current-current correlation functions, we first of all have to
fix some notation. In the antiferromagnetic massive regime,
the dispersion relation of the elementary excitations can be
expressed explicitly in terms of theta functions
\begin{subequations}
\label{momen}
\begin{align} \label{ptheta4}
      p(\la) & = \frac{\p}{2} + \la
             - \i \ln \biggl(
	       \frac{\dh_4 (\la + \i \g/2| q^2)}{\dh_4 (\la - \i \g/2| q^2)}
	       \biggr) \epc \\[1ex] \label{dressede}
     \e(\la) & = - 2 J \sh (\g) \dh_3 \dh_4
                       \frac{\dh_3 (\la)}{\dh_4 (\la)} \epp
\end{align}
\end{subequations}
Here $p$ is the momentum, $\e$ is the dressed energy (for
$h = 0$), $\lambda$ the rapidity, and we will use the convention $\dh_j=\dh_j(0|q)$.

The integrands in each term of our form factor series are
parameterized in terms of two sets ${\cal U} = \{u_j\}_{j=1}^\ell$
and ${\cal V} = \{v_k\}_{k=1}^\ell$ of `hole and particle type'
rapidity variables of equal cardinality $\ell$. For sums and
products over these variables we shall employ the short-hand
notations
\begin{equation} \label{defsetsumprod}
     \sum_{\la \in {\cal U} \ominus {\cal V}} f(\la)
        = \sum_{\la \in {\cal U}} f(\la)
          - \sum_{\la \in {\cal V}} f(\la) \epc \qd
     \prod_{\la \in {\cal U} \ominus {\cal V}} f(\la)
        = \frac{\prod_{\la \in {\cal U}} f(\la)}{\prod_{\la \in {\cal V}} f(\la)} \epp
\end{equation}
We define
\begin{equation}
     \Si = - \frac{\p k}{2} - \2 \sum_{\la \in {\cal U} \ominus {\cal V}} \la
\end{equation}
and
\begin{equation} \label{defphpm}
     \Ph^{(\pm)} (\la) =
        \re^{\pm \i \Si}
        \prod_{\m \in {\cal U} \ominus {\cal V}}
        \tst{\G_{q^4} \bigl(\2 \pm \frac{\la - \m}{2 \i \g}\bigr)
             \G_{q^4} \bigl(1 \mp \frac{\la - \m}{2 \i \g}\bigr)} \epp
\end{equation}
We introduce multiplicative spectral parameters $H_j = \re^{2 \i u_j}$,
$P_k = \re^{2 \i v_k}$ and the following special basic hypergeometric
series,
\begin{subequations}
\begin{align}
     \PH_1 (P_k, \a) & = \:
        _{2 \ell} \PH_{2\ell - 1}
           \Biggl(\begin{array}{@{}r} q^{-2},
                     \{q^2 \frac{P_k}{P_m}\}_{m \ne k}^\ell,
                     \{\frac{P_k}{H_m}\}_{m=1}^\ell \\[.5ex]
                     \{\frac{P_k}{P_m}\}_{m \ne k}^\ell,
                     \{q^2 \frac{P_k}{H_m}\}_{m=1}^\ell
                  \end{array}; q^4, q^{4 + 2 \a}
           \Biggr), \\[1ex]
     \PH_2 (P_k, P_j, \a) & = \:
        _{2 \ell} \PH_{2\ell - 1}
           \Biggl(\begin{array}{@{}r} q^6, q^2 \frac{P_j}{P_k},
                     \{q^6 \frac{P_j}{P_m}\}_{m \ne k, j}^\ell,
                     \{q^4 \frac{P_j}{H_m}\}_{m=1}^\ell \\[.5ex]
                     q^8 \frac{P_j}{P_k},
                     \{q^4 \frac{P_j}{P_m}\}_{m \ne k, j}^\ell,
                     \{q^6 \frac{P_j}{H_m}\}_{m=1}^\ell
                  \end{array}; q^4, q^{4 + 2 \a}
           \Biggr) \epp
\end{align}
% JS: How are \frac{P_k}{H_m}\}_m^\ell defined?
\end{subequations}
We further define
\begin{equation}
     \Ps_2 (P_k, P_j, \a) = q^{2 \a} r_\ell (P_k, P_j) \PH_2 (P_k, P_j, \a) \epc
\end{equation}
where
\begin{equation}
     r_\ell (P_k, P_j) = \frac{q^2 (1 - q^2)^2 \frac{P_j}{P_k}}
                              {(1 - \frac{P_j}{P_k})(1 - q^4 \frac{P_j}{P_k})}
        \Biggl[\prod_{\substack{m=1 \\ m \ne j, k}}^\ell
               \frac{1 - q^2 \frac{P_j}{P_m}}{1 - \frac{P_j}{P_m}}\Biggr]
        \Biggl[\prod_{m=1}^\ell \frac{1 - \frac{P_j}{H_m}}
	                             {1 - q^2 \frac{P_j}{H_m}}\Biggr] \epc
\end{equation}
and introduce a `conjugation' $\overline f (H_j, P_k, q^\a)
= f(1/H_j, 1/P_k, q^{- \a})$.

This allows to define the core part of our form factor
densities, which is a matrix ${\cal M}$ with matrix elements
\begin{multline}
     {\cal M}_{i, j} =
        \de_{ij} \biggl[\overline{\PH}_1 (P_j, 0) -
                        \frac{\Ph^{(-)} (v_j)}{\Ph^{(+)} (v_j)} \PH_1 (P_j, 0)\biggr] \\
        - (1 - \de_{ij})\biggl[\overline{\Ps}_2 (P_j, P_i, 0) -
                               \frac{\Ph^{(-)} (v_i)}{\Ph^{(+)} (v_i)}
                               \Ps_2 (P_j, P_i, 0)\biggr] \epp
\end{multline}
By $\hat {\cal M}$ we denote the matrix obtained from
${\cal M}$ upon replacing $u_j \leftrightharpoons - v_j$.
We finally introduce one more function
\begin{equation}
     \Xi (\la) = \frac{\G_{q^4} \bigl(\2 + \frac{\la}{2 \i \g}\bigr)
	          G_{q^4}^2 \bigl(1 + \frac{\la}{2 \i \g}\bigr)}
	         {\G_{q^4} \bigl(1 + \frac{\la}{2 \i \g}\bigr)
	          G_{q^4}^2 \bigl(\2 + \frac{\la}{2 \i \g}\bigr)} \epp
\end{equation}
Then the form factor amplitudes of the current-current correlation
functions are
\begin{multline}\label{eq:Azz}
     {\cal A}_{\cal J}^{(2 \ell)} ({\cal U}, {\cal V}|k) =
        \biggl(\frac{\sum_{\la \in {\cal U} \ominus {\cal V}} \e(\la)}
	            {4 \dh_1(\Si)/\dh_1'} \biggr)^2
	\biggl[
	   \prod_{\la, \m \in {\cal U} \ominus {\cal V}}
	   \Xi (\la - \m)
        \biggr] \\ \times
	\det_\ell \{{\cal M}\} \det_\ell \{\hat {\cal M}\}
	\det_\ell \biggl(\frac{1}{\sin(u_j - v_k)}\biggr)^2 
%JS: Should the u_j be a u_\ell ?
\end{multline}
where $\e(\la)$ is the dressed energy and $\dh_1' = \dh_1'(0|q)$.

Using these amplitudes as well as the momentum and dressed energy
defined in (\ref{momen}) we can formulate our main result. For
every $m \ge 0$, the dynamical two-point function of spin currents
of the XXZ chain in the antiferromagnetic massive regime and
in the zero-temperature limit can be represented by the form-factor
series
\begin{equation} \label{curcur}
     \bigl\< {\cal J}_1 (t) {\cal J}_{m+1} \bigr\> =
        \sum_{\ell = 1}^\infty C^{(2\ell)} (m, t) \epc
\end{equation}
where
\begin{equation} \label{lplh}
     C^{(2\ell)} (m, t) =
        \sum_{k = 0, 1} \frac{(-1)^{mk}}{(\ell !)^2}
	   \int_{{\cal C}_h^\nex} \frac{\rd^\nex u}{(2\p)^\nex}
	   \int_{{\cal C}_p^\nex} \frac{\rd^\nex v}{(2\p)^\nex} \:
	   {\cal A}_{\cal J}^{(2 \ell)} ({\cal U}, {\cal V}|k)
	   \re^{\i \sum_{\la \in {\cal U} \ominus {\cal V}} (t \e (\la) - m p(\la))}
\end{equation}
is the $\ell$-particle $\ell$-hole contribution. The integration
contours in (\ref{lplh}) can be chosen as ${\cal C}_h =
[-\frac{\p}{2}, \frac{\p}{2}] - \frac{\i \g}{2} + \i \de$
and ${\cal C}_p = [-\frac{\p}{2}, \frac{\p}{2}] +
\frac{\i \g}{2} + \i \de'$, where $\de, \de' > 0$ are small.
The derivation of (\ref{curcur}), (\ref{lplh}) is slightly
cumbersome. It relies on a generalization of our work in
Ref.~\cite{GKKKS17} that will be published elsewhere and
on the technical achievements obtained in Refs.~\cite{BGKS21a,BGKSS21}.

%%%%%%%%%%%%%%%%%%%%%%%%%%%%%%%%%%%%%%
\subsection{Spin-current correlations}
\label{sec:sccorr}
%%%%%%%%%%%%%%%%%%%%%%%%%%%%%%%%%%%%%%
Except for the vicinity of the isotropic point, the contributions
from higher $\ell$ terms in (\ref{curcur}) to $\bigl\< {\cal J}_1 (t)
{\cal J}_{m+1} \bigr\>$ turn out to be small and to decrease rapidly
in time. We plot the contributions of the $\ell=1$ term and of the
sum of the $\ell=1$  and $\ell=2$ terms to the current autocorrelation
function ($m=0$) for $\Delta=1.2$ in Fig.~\ref{fig:JJm0D12}.
\begin{figure}[!h]
\centering
\includegraphics[width=0.98\columnwidth]{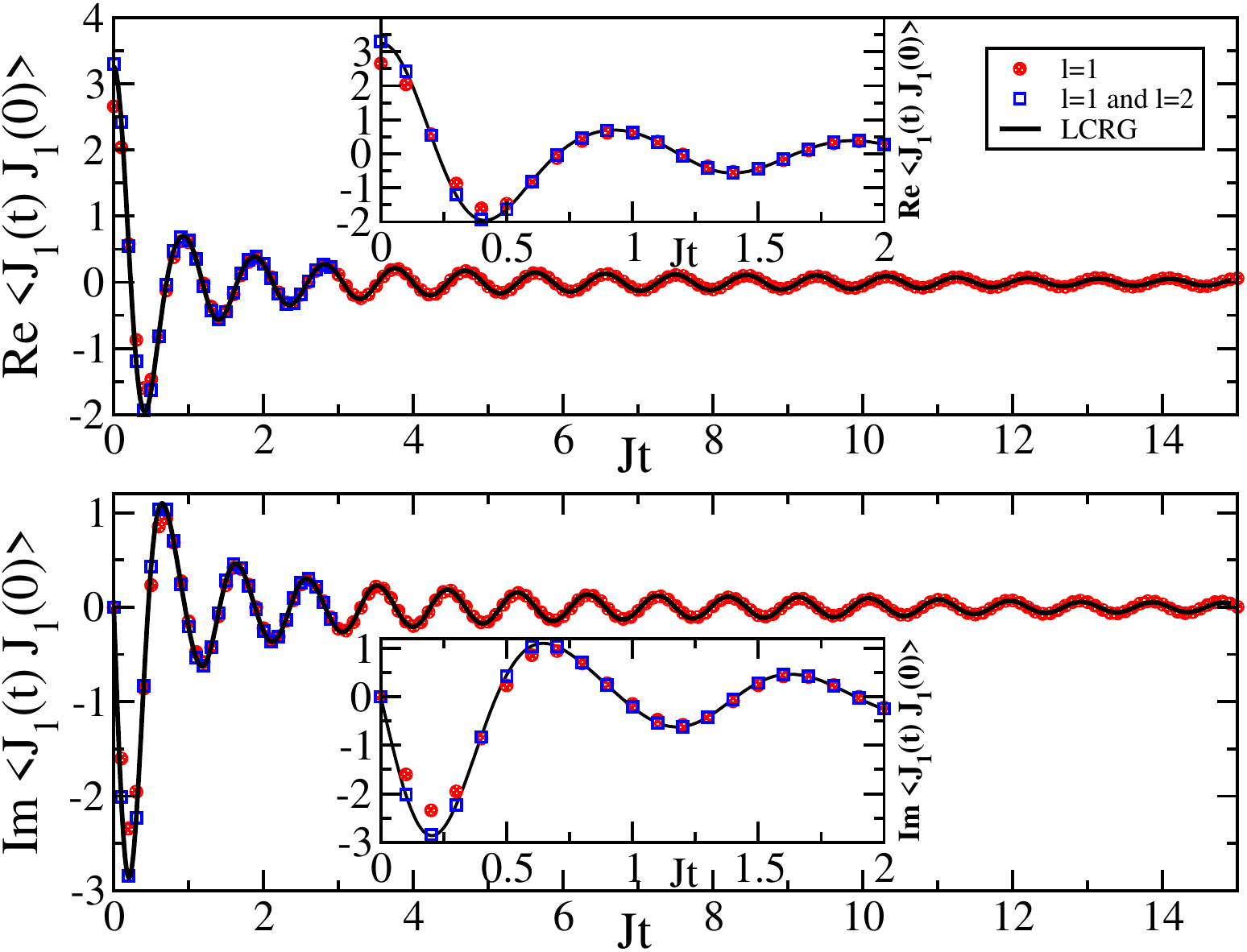}
\caption{The real and imaginary contributions of the $\ell=1$ term
and of the sum of $\ell=1$ and $\ell=2$ terms to $\bigl\< {\cal J}_1 (t)
{\cal J}_{1} \bigr\>$ for $\Delta=1.2$ are compared to LCRG results.
The insets show that the $\ell=2$ contribution becomes negligible on this
scale for $tJ>2$.}
\label{fig:JJm0D12}
\end{figure}

In order to estimate the truncation error of our exact series
representation we compare with independent exact results for
$t = 0$. For small $m$ such results
are available due to the factorization of the reduced density
matrix in the static case \cite{BGKS07,BJMST08a,JMS08}. We have
checked that, for $0\le m \le 2$ with various values of $\Delta$,
away from the isotropic point, the sum of the $\ell=1$ and $\ell=2$
terms in (\ref{curcur}) recovers the known exact values with 
good accuracy. For example, for $\Delta=1.5$, the exact value
of $\bigl\< {\cal J}_1 (0) {\cal J}_{2} \bigr\>$ is $-0.333748\ldots$,
while (\ref{curcur}) yields $-0.333687\ldots$.

%\textcolor{red}{
%The exact knowledge of the spin-current correlation function is restricted
%to the case with equal time $t=0$ and with small $m$. We have checked
%that, for $0\le m \le 2$ with various values of $\Delta$, away from the isotropic  point, 
%the sum of $\ell=1$ and $\ell=2$ term recovers the known exact values in good accuracy. 
%For example, for $\Delta=1.5$, the exact value of  $\bigl\< {\cal J}_1 (0)
%{\cal J}_{2} \bigr\>$ is $-0.333748\cdots$, while  (\ref{curcur})  yields  $-0.333687\cdots$.
%}

For $t>0$ independent exact data are no longer available and the
results of the form factor series are compared to a light-cone
renormalization group (LCRG) calculation. The latter is a
density-matrix renormalization group algorithm which makes use
of the Lieb-Robinson bounds to obtain results for infinite 
chain lengths \cite{EnSi12,EAS17}. In the LCRG calculations,
we keep $8192$ states in the truncated Hilbert space. The
truncation error reaches $\sim 10^{-6}$ at the longest simulation
times shown. A comparison with calculations where the number of
states kept is varied (not shown) suggests that the error of
the LCRG calculations remains always smaller than the size of
the symbols used to represent the results from the form factor
series. The LCRG data and the form factor series are in excellent
agreement.  

The explicit formula (\ref{curcur}) makes it possible to evaluate
the correlation function $\<{\cal J}_1 (t) {\cal J}\bigr\> \equiv
\sum_{m=0}^{L_c} \<{\cal J}_1 (t) {\cal J}_{m+1}\bigr\>$, where
$L_c \in {\mathbb N}$, even for large $L_c$ and $t$ numerically.
As an illustration, we plot the real and imaginary parts of
$\<{\cal J}_1 (t) {\cal J}\bigr\>$ for $\Delta=3$ with $L_c=349$
in Figure~\ref{fig:JJD3L1}. The contributions from higher $\ell$
are small in this case and we only include $\ell=1$. The high-frequency
oscillation in the right plot, centred around zero, and the continuous
decay of the correlation function for long times clearly indicate
that, in agreement with common belief, the zero-temperature spin
transport is non-ballistic. A ballistic contribution would show
up as a non-vanishing constant long-time asymptotics of this
correlation function.
\begin{figure}[!h]
\centering
{\includegraphics[width=.48\textwidth]{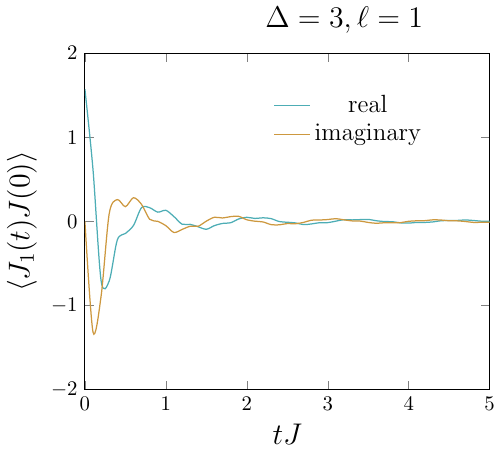} \hfill
\includegraphics[width=.48\textwidth]{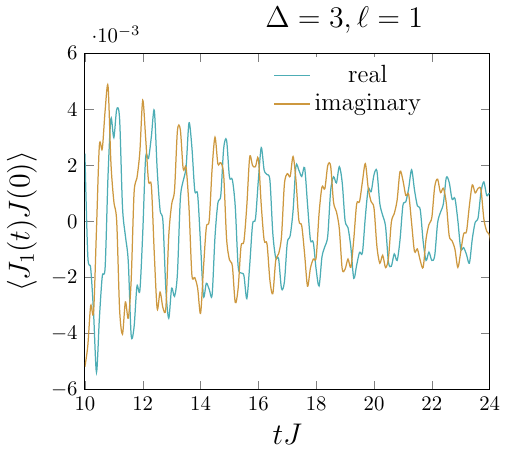}}
\caption{$\<{\cal J}_1 (t) {\cal J}\bigr\>$ for $\Delta=3$, $L_c=349$
and times $0<tJ<5$ (left), $10 < tJ <24$ (right).
}
\label{fig:JJD3L1}
\end{figure}

For smaller $\Delta$, the conclusion is less obvious from our data
for small times, see the left panel in Fig.~\ref{fig:JJD15L1}.
However, in contrast to purely numerical methods the thermal form factor
expansion allows to obtain highly accurate data for long times,
see the right panel in Fig.~\ref{fig:JJD15L1}. These data clearly
show that the correlation function decays and has low-frequency
oscillations.
\begin{figure}[!h]
\centering
{\includegraphics[height=6.4cm]{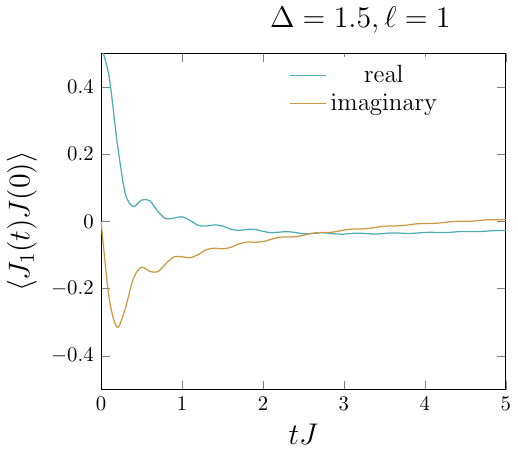} \hfill
\includegraphics[height=6.4cm]{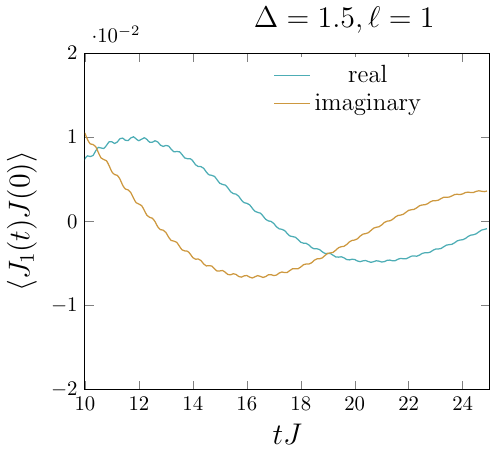}}
\caption{$\<{\cal J}_1 (t) {\cal J}\bigr\>$ for $\Delta=1.5$ and
times  $0<tJ<5$ (left), $10 < tJ <24$ (right). We used $L_c=219$.
}
\label{fig:JJD15L1}
\end{figure}

The upper limit  $L_c$ in the sum over distances between the current
operators is determined by a characteristic velocity of the excitations
and the maximum time scale we want to reach. We fix $L_c \gtrsim v_2 t$,
where $v_2$ is the upper critical velocity appearing in the
long-time large-distance asymptotic analysis of the current-current
correlation functions (see Sec.~\ref{Sec:asyana} below).
Typical values of $v_2$ are listed in Table \ref{tab_velocity}.
\begin{table}[!h]
\begin{center}
\setlength{\tabcolsep}{1em}
\begin{tabular}{cccc}
\toprule
    $\Delta$ & $1.5$ &  $2$ & $3$ \\ \midrule[1pt]
    $v_2/J$ & $7.47329$ & $9.06159$ & $12.6851$ \\
\bottomrule
\end{tabular}
\caption{Velocity $v_2/J$ for various anisotropies
$\Delta$ according to equation (\ref{critvs}).}
\label{tab_velocity}
\end{center}
\end{table}
We will choose similar values of $L_c$ in the sections below unless
the contributions from large distances $m$ are negligible.

%%%%%%%%%%%%%%%%%%%%%%%%%%%%%%%%%%%%%%%%%%%%%%%%
\subsection{Long-time large-distance asymptotics}
\label{Sec:asyana}
%%%%%%%%%%%%%%%%%%%%%%%%%%%%%%%%%%%%%%%%%%%%%%%%
From the example in Fig.~\ref{fig:JJm0D12}, we see that the
two-particle two-hole term significantly contributes to the numerical
value of the correlation function only at short times. The long-time
large-distance asymptotics of the correlation function is entirely
determined by the one-particle one-hole term,
\begin{equation}
     \bigl\< {\cal J}_1 (t) {\cal J}_{m+1} \bigr\> \sim C^{(2)} (m, t) \epp
\end{equation}
This is what makes the series representation (\ref{curcur}),
(\ref{lplh}) so efficient. The double integral $C^{(2)} (m, t)$
can be numerically evaluated as accurately as we wish, because
its asymptotic behaviour for $m, t \rightarrow \infty$ at fixed
ratio $v = m/t$ is known in closed form from a saddle-point analysis.
Such type of analysis was carried out for the two-point functions of the
local magnetization $\<\s_1^z (t) \s_{m+1}^z\>$ in one of our previous
works \cite{DGKS16a}. Here the mathematical problem is exactly the
same. Referring to equations (\ref{ctwosym}), (\ref{atwosym}) in
Appendix~\ref{app:twospincond}, we can rewrite $C^{(2)} (m, t)$ as
\begin{equation} \label{c2fg}
     C^{(2)} (m, t) =
        \int_{- \frac{\p}{2}}^\frac{\p}{2} \frac{\rd z_1}{2\p}
	\int_{- \frac{\p}{2}}^\frac{\p}{2} \frac{\rd z_2}{2\p} \:
	   f(z_1, z_2) \re^{t (g(z_1) + g(z_2))} \epc
\end{equation}
where
\begin{subequations}
\begin{align}
     g (z) & = \i \bigl(\e(z) - v p(z)\bigr) \epc \\[1ex]
     f(z_1, z_2) & = {\cal A}^{(2)}_s (z_1, z_2|0)
                     + (-1)^m {\cal A}^{(2)}_s (z_1, z_2|1) \epp
\end{align}
\end{subequations}
The definition of ${\cal A}^{(2)}_s$ can be found in equation
(\ref{atwosym}) below. The important point is that $f(z_1, z_2)$
has a double zero at $z_1 = z_2$ implying that (\ref{c2fg}) is
of the same form as the integral analysed in \cite{DGKS16a}.

Let us briefly recall the main results of \cite{DGKS16a}. The
asymptotics of $C^{(2)}$ is determined by the roots of the saddle-%
point equation $g' (z) = 0$ on steepest descent contours joining
$- \p/2$ and $\p/2$. The saddle-point equation is most compactly
expressed in terms of Jacobi elliptic functions and their parameters
that will also occur below in our discussion of the one-particle
one-hole contribution to the optical conductivity. We shall need
the elliptic module~$k$, the complementary module $k'$ and the
complete elliptic integral $K$. They are all conveniently
parameterized in terms of the elliptic nome $q$ by means of
$\dh_j\equiv\dh_j(0|q)$,
\begin{equation} \label{ellipticcons}
     k = \dh_2^2/\dh_3^2 \epc \qd
     k' = \dh_4^2/\dh_3^2 \epc \qd
     K = \p \dh_3^2/2 \epp
\end{equation}
Let
\begin{equation} \label{critvs}
     v_1 = \frac{4J K k^2 \sh(\g)}{\p (1 + k')} \epc \qd
     v_2 = \frac{4J K k^2 \sh(\g)}{\p (1 - k')}
\end{equation}
and $k_1 = v_1/v_2$. The first relation (\ref{ellipticcons}) is
invertible which allows us to interpret $K$ as a function of $k$.
Let $K_1 = K(k_1)$. Then the saddle-point equation can be
represented as
\begin{equation} \label{spe}
     \sn \biggl(\frac{4 K_1 z}{\p}\bigg|k_1\biggr) = \frac{v}{v_1} \epc
\end{equation}
where $\sn$ is a Jacobi elliptic function. The solutions of the
saddle-point equation divide the $m$-$t$ world plane into three
different asymptotic regimes,
\begin{equation}
     0 < v < v_1 \epc \qd v_1 < v < v_2 \epc \qd v_2 < v \epc
\end{equation}
which were called \cite{DGKS16a} the `time-like regime', the `precursor
regime' and the `space-like regime' by analogy with the asymptotic
analysis of electro-magnetic wave propagation in media.

Here we recall only the result of the asymptotic analysis in the
time-like regime as it is relevant for the `true long-time behaviour'.
For the other two asymptotic regimes the reader is referred to
\cite{DGKS16a}. In the time-like regime, $0 < v/v_1 < 1$, the
saddle-point equation (\ref{spe}) has two inequivalent real
solutions,
\begin{equation}
     \la^- = \frac{\p}{4 K_1}
             \arcsn \biggl(\frac{v}{v_1}\bigg|k_1\biggr) \epc \qd
	     \la^+ = \frac{\p}{2} - \la^- \epc
\end{equation}
located in the interval $[0, \p/2]$. The function occurring on
the right hand side is the inverse Jacobi-$\sn$ function. The
long-time large-distance asymptotics of $C^{(2)} (m, t)$ in
the time-like regime is then determined by the saddle points,
\begin{equation}
     C^{(2)} (m, t) \sim \frac{f(\la^+, \la^-)}{\p t}
        \prod_{\s = \pm} \frac{\re^{t g(\la^\s)}}{\sqrt{g''(\la^\s)}} \epp
\end{equation}
Note that the product on the right hand side can be expressed
explicitly in terms of elementary transcendental functions of
$v = m/t$ \cite{DGKS16a}.

%%%%%%%%%%%%%%%%%%%%%%%%%%%%%%%%%%
%%%%%%%%%%%%%%%%%%%%%%%%%%%%%%%%%%
\section{Optical Conductivity}
\label{Sec3}
%%%%%%%%%%%%%%%%%%%%%%%%%%%%%%%%%%
%%%%%%%%%%%%%%%%%%%%%%%%%%%%%%%%%%
Quite generally, current-current correlation functions determine
transport coefficients within the framework of linear response
theory. The correlation function of two spin-current operators
considered above determines the optical spin-conductivity $\s (\om)$.
\subsection{Form factor series}
Several equivalent formula expressing $\sigma(\omega)$ for the
XXZ chain in terms of current-current correlation functions have
been described in the literature (see e.g.\ \cite{Sirker20,BHKPSZ21}).
For our convenience and in order to make this work more
self-contained, we have included concise derivations in Appendix~\ref{app:spinconduct}.
%\ref{app:spinconduct}.
We are interested in the thermodynamic limit in which the following
lemma holds true.
\begin{lemma} \label{lem:redyncond}
The real part of the optical spin conductivity of the XXZ chain
can be represented as
\begin{equation} \label{redyncond}
     \Re \s (\om) = \frac{1 - \re^{- \frac{\om}{T}}}{2 \om}
        \int_{- \infty}^\infty \rd t \: \re^{\i \om t}
	\biggl(
        2 \sum_{m = 0}^\infty \bigl\< {\cal J}_1 (t) {\cal J}_{m+1} \bigr\>_T
                          - \bigl\< {\cal J}_1 (t) {\cal J}_1 \bigr\>_T
			    \biggr) \epp 
\end{equation}
\end{lemma}
\begin{proof}
We start with equation (\ref{resigfluctform}) from
Appendix~\ref{app:spinconduct}, 
\begin{equation}\label{Resigma_L1}
     \Re \s (\om) = \frac{1 - \re^{- \frac{\om}{T}}}{2 \om}
        \int_{- \infty}^\infty \rd t \: \re^{\i \om t}
        \lim_{L \rightarrow \infty} \frac{1}{L}
	\sum_{j, m = 1}^L \bigl\<{\cal J}_j (t) {\cal J}_m\bigr\>_T \epp
\end{equation}
Here $L$ is the length of the periodic system. We shall assume $L$
to be even. The Hamiltonian is invariant under translations modulo
$L$ and under parity transformations $j \rightarrow L - j + 1$. Hence,
\begin{multline} \label{calcjjlim}
     \frac{1}{L} \sum_{j, m = 1}^L \bigl\<{\cal J}_j (t) {\cal J}_m\bigr\>_T
        = \sum_{m = 1}^L \bigl\<{\cal J}_1 (t) {\cal J}_m\bigr\>_T \\[-1ex] =
          \bigl\<{\cal J}_1 (t) {\cal J}_1\bigr\>_T
	  + \sum_{m = 1}^{\frac{L}{2} - 1} \bigl\<{\cal J}_1 (t) {\cal J}_{m + 1} \bigr\>_T
          + \bigl\<{\cal J}_1 (t) {\cal J}_{\frac{L}{2} + 1} \bigr\>_T
	  + \sum_{m = \frac{L}{2} + 1}^{L - 1}
	       \bigl\<{\cal J}_1 (t) {\cal J}_{m + 1} \bigr\>_T \\[-2ex] =
          \bigl\<{\cal J}_1 (t) {\cal J}_1\bigr\>_T
	  + 2 \sum_{m = 1}^{\frac{L}{2} - 1} \bigl\<{\cal J}_1 (t) {\cal J}_{m + 1} \bigr\>_T
          + \bigl\<{\cal J}_1 (t) {\cal J}_{\frac{L}{2} + 1} \bigr\>_T \epp
\end{multline}
Here, we have split the summation so that, upon using the $L$-periodicity
of the lattice, the summed-up terms do get farther and farther away from
the first site. This produces the factor of $2$ and is necessary for
appropriately taking the thermodynamic limit. Hence, 
\begin{equation}
     \lim_{L \rightarrow \infty}
        \frac{1}{L} \sum_{j, m = 1}^L \bigl\<{\cal J}_j (t) {\cal J}_m\bigr\>_T =
          \bigl\<{\cal J}_1 (t) {\cal J}_1\bigr\>_T
	  + 2 \sum_{m = 1}^\infty \bigl\<{\cal J}_1 (t) {\cal J}_{m + 1} \bigr\>_T \epc
\end{equation}
where the expectation values on the right hand side are now those
in the thermodynamic limit ($m$ fixed, $L \rightarrow \infty$).
\end{proof}

The zero-temperature limit of (\ref{redyncond}) for $\om > 0$ is
obvious. In this limit we can take our numerical results for
$\<{\cal J}_1 (t) {\cal J}\>$ from section \ref{sec:sccorr} with
$L_c$ sufficiently large in order to compute $\Re \s (\om)$
numerically from Lemma~\ref{lem:redyncond}. The results shown
in Fig.~\ref{fig:TwosipnonconductivityD3} and
Fig.~\ref{fig:conductivityD3L12} below were obtained this way. On
the other hand, the summation over all lattice sites involved in
(\ref{redyncond}) can be easily carried out analytically on the
series representation (\ref{curcur}), (\ref{lplh}), and we obtain
the following lemma.

\begin{lemma} \label{lem:cckzero}
In the antiferromagnetic massive regime for $T \rightarrow 0$ the
correlation function under the integral in (\ref{redyncond}) has
the thermal form factor series representation
\begin{multline} \label{curcurkzero}
     2 \sum_{m = 0}^\infty \bigl\< {\cal J}_1 (t) {\cal J}_{m+1} \bigr\>
                          - \bigl\< {\cal J}_1 (t) {\cal J}_1 \bigr\> = \\[-1ex]
        \sum_{\substack{\nex \in {\mathbb N}\\k = 0, 1}}
	  \frac{1}{(\nex !)^2}
	     \int_{{\cal C}_h^\nex} \frac{\rd^\nex u}{(2\p)^\nex}
	     \int_{{\cal C}_p^\nex} \frac{\rd^\nex v}{(2\p)^\nex} \:
	     {\cal A}_\s^{(2 \ell)} ({\cal U}, {\cal V}|k)
	     \re^{\i t \sum_{\la \in {\cal U} \ominus {\cal V}} \e (\la|h)} \epc
\end{multline}
where
\begin{equation} \label{ampsigma}
     {\cal A}_\s^{(2 \ell)} ({\cal U}, {\cal V}|k) =
        - \i \ctg \Bigl(\tst{\2 \bigl(\p k +
	        \sum_{\la \in {\cal U} \ominus {\cal V}} p(\la)\bigr)}\Bigr)
		{\cal A}_{\cal J}^{(2 \ell)} ({\cal U}, {\cal V}|k) \epp
\end{equation}
\end{lemma}
\begin{proof}
$\Re \sum_{\la \in {\cal U} \ominus {\cal V}} \i p(\la) > 0$ if
$u_j \in {\cal C}_h$ and $v_k \in {\cal C}_p$ as can, for instance,
be seen from Appendix~A.2 of \cite{DGKS15a}. Hence, the summation over
$m$ can be performed by the geometric sum formula.
\end{proof}
With equation (\ref{curcurkzero}), we have an alternative starting point
for a numerical computation of the real part of the optical conductivity
at $T=0$. Again, we would have to substitute this formula into
(\ref{redyncond}) for $T = 0$, $\om > 0$ and calculate the Fourier
transform numerically. For the time being we refrain from this possibility.
The direct use of (\ref{curcur}) and (\ref{lplh}), for which we could
resort to existing computer programs \cite{BGKSS21}, gives reliable
numerical results as can be seen from Fig.~\ref{fig:TwosipnonconductivityD3}.
Our algorithm uses numerical saddle point integration. For a use
with (\ref{curcurkzero}), (\ref{ampsigma}) we would have to modify it
as some of the poles of the integrand are now located at the saddle
points. We leave the interesting questions of how to deal with this
situation numerically and how to calculate the long-time asymptotics
analytically for future studies.

%%%%%%%%%%%%%%%%%%%%%%%%%%%%%%%%%%%%%%%%%%%
\subsection{Two-spinon contribution}
%%%%%%%%%%%%%%%%%%%%%%%%%%%%%%%%%%%%%%%%%%%
In the general case of $\ell$ particle-hole excitations, the numerical
evaluation of the Fourier transform seems to be the most efficient
way to make use of Lemma~\ref{lem:cckzero}. For the $\ell=1$ particle-hole
contribution, however, following the examples of \cite{BKM98,CMP08,%
BGKKW12}, the Fourier transformation can be carried out by hand.
The details of the calculation are discussed in
Appendix~\ref{app:twospincond}.

We introduce two functions
\begin{subequations}
\label{randb}
\begin{align}
    r(\om) & = \frac{\p}{K}
               \arcsn \biggl(\frac{\sqrt{(h_\ell/k')^2 - \om^2}}
	                          {h_\ell k/k'}\bigg|k \biggr) \epc
				  \label{eq:def_r} \\[1ex]
     B(z) & =
        \frac{1}{G_{q^4}^4 \bigl(\2\bigr)}
	\prod_{\s = \pm}
	\frac{G_{q^4} \bigl(\tst{1 + \frac{\s z}{2 \i \g}}\bigr)
	      G_{q^4} \bigl(\tst{\frac{\s z}{2 \i \g}}\bigr)}
	     {G_{q^4} \bigl(\tst{\frac{3}{2} + \frac{\s z}{2 \i \g}}\bigr)
	      G_{q^4} \bigl(\tst{\2 + \frac{\s z}{2 \i \g}}\bigr)} \epc
	      \label{eq:def_B} 
\end{align}
\end{subequations}
where $h_\ell$ has been defined at the beginning of Sec.~\ref{FF}.
Using (\ref{ellipticcons}) and (\ref{randb}) we can formulate our
result for the two-spinon contribution to the optical conductivity.

\begin{lemma} \label{lem:twoconduct}
The two-spinon contribution to the real part of the optical
conductivity of the XXZ chain at zero temperature and in the
antiferromagnetic massive regime can be represented as
\begin{equation}\label{sigmaTwospinon}
     \Re \s^{(2)} (\om) = \frac{q^\2 h_\ell^2 k}{8 k'}
        \frac{B\bigl(r (\om)\bigr)}{\D - \cos\bigl(r(\om)\bigr)}
	\frac{\dh_3^2}{\dh_3^2 \bigl(r(\om)/2\bigr)}
	\frac{1}{\sqrt{\bigl((h_\ell/k')^2 - \om^2\bigr)
	                 \bigl(\om^2 - h_\ell^2\bigr)}} \epc
\end{equation}
as long as $h_\ell < \omega < h_\ell/k'$. It vanishes outside
of this range of $\omega$. 
\end{lemma}
The derivation of this result is discussed in Appendix \ref{app:twospincond}. In the above expression (\ref{sigmaTwospinon}), we identify two
van-Hove singularities at the upper and lower 2-spinon band edges,
see the last factor on the right hand side. They are both
canceled by $B\bigl(r (\om)\bigr)$ as $B(z)$ has double zeros
at integer multiples of $\pi$. As a result, $\Re \s^{(2)} (\om)$
has square-root singularities, both at the lower and the upper
2-spinon band edges, away from the isotropic point.

As a consistency test, we plot the spin conductivity obtained
from a numerical Fourier transformation of the $\ell=1$ parts
of $\bigl\< {\cal J}_1 (t) {\cal J}_{m+1} \bigr\>$ and the above
analytic result in Fig.~\ref{fig:TwosipnonconductivityD3}. The
two curves agree well except for very small frequencies $\om$.
%\textcolor{blue}{\sout{,
%where inaccuracies in the numerical Fourier transform are enhanced.}}
Note that we perform the summation and the numerical integration in
(\ref{Resigma_L1}) without introducing any window functions or filters.
The overall factor $\omega^{-1}$ in (\ref{Resigma_L1}) could then
introduce an artificial instability. We nevertheless observe only a
small deviation from zero in the vicinity of $\om=0$. This indicates a
high accuracy of our spin-current correlation data.

\begin{figure}[!h]
\centering
{\includegraphics[width=8cm]{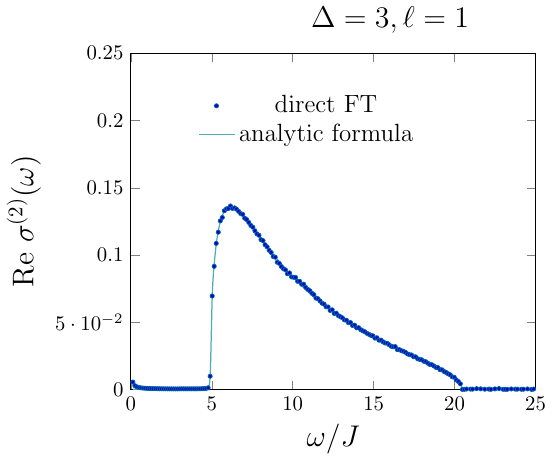}}
\caption{Comparison of the analytic result (\ref {sigmaTwospinon})
and a numerical Fourier transformation of the $\ell=1$ part of
$\bigl\< {\cal J}_1 (t) {\cal J}_{m+1} \bigr\>$ for anisotropy
$\Delta=3$. For the latter we use $0\le m \le 399$ and $0\le tJ \le 50$.
}
\label{fig:TwosipnonconductivityD3}
\end{figure}

With decreasing anisotropy $\Delta$, the peak position moves towards
$\omega=0$ while its height increases, see Fig.~\ref{fig:SigmaL1}.
\begin{figure}[!h]
\centering
{\includegraphics[width=8cm]{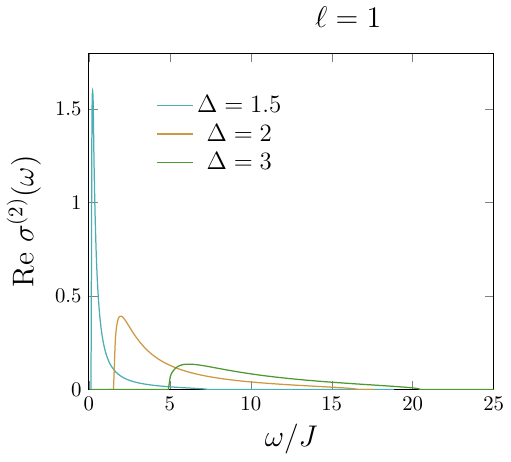}}
\caption{$\ell=1$ contribution, Eq.~(\ref{sigmaTwospinon}), to ${\rm Re}\, \sigma^{(2)}(\om)$ for various $\Delta$.}
\label{fig:SigmaL1}
\end{figure}
This is expected because the XXZ chain has a non-zero $T=0$
Drude weight in the isotropic limit \cite{ShSu90}. A short
discussion of the isotropic limit is presented in Appendix \ref{app:twospincond}.

%%%%%%%%%%%%%%%%%%%%%%%%%%%%%%%%%%%%%%%%%%%%%%%%%%%%
\subsection{More than two spinons}\label{sec:4spinons}
%%%%%%%%%%%%%%%%%%%%%%%%%%%%%%%%%%%%%%%%%%%%%%%%%%%%
For $\ell \ge 2$, simple analytic expressions like (\ref{sigmaTwospinon})
are not available. Nevertheless, as has already been mentioned, the
spin-current correlation function is enumerable for sufficiently large
$m$ and $t$. This enables us to perform a direct Fourier transformation,
at least in principle. In practice, calculating higher contributions is
time-consuming. Here we therefore only present a result for $\ell=2$,
which corresponds to the 4-spinon case. Inside the two-spinon band
the additional contribution from 4-spinon states is small, away from the
isotropic point. Above the 2-spinon upper band edge, however,
a small conductivity is entirely carried by the scattering states of
four and more spinons. For $\Delta=2,3$ the maximum of the $\ell=2$
contribution is located close to the upper 2-spinon band edge,
and we expect this contribution to be significant even near the
lower 2-spinon band edge as $\Delta \rightarrow 1$. 
An example for $\Delta=3$ is shown in Fig.~\ref{fig:conductivityD3L12}.
 \begin{figure}[!h]
\centering
{\includegraphics[width=8cm]{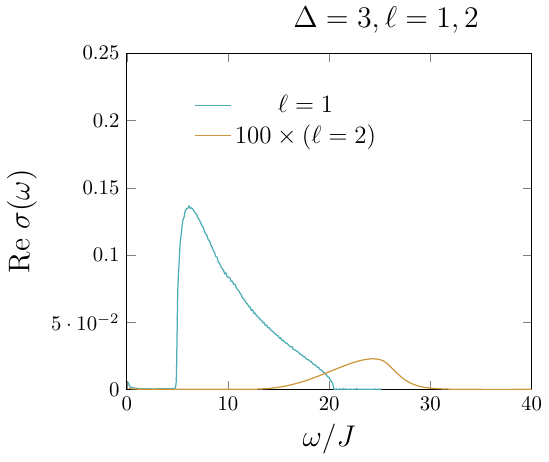}}
\caption{The $\ell = 1$ and $\ell = 2$ contributions to ${\rm Re}\,
\sigma(\om)$ for $\Delta=3$. Note that the contribution from $\ell=2$
is multiplied by a factor 100. For the Fourier transform for $\ell=2$
we use data for $\bigl\< {\cal J}_1 (t) {\cal J}_{m+1} \bigr\>$ with
$0\le m \le 39$ and $0 \le tJ \le 30$.}
\label{fig:conductivityD3L12}
\end{figure}

As a benchmark for the accuracy of our results we consider the $f$-sum
rule \cite{BariAdler},
\begin{equation}\label{sumrule}
     \int_0^\infty \rd \om \: \Re \s(\om) =
        - \lim_{L\rightarrow \infty}  \frac{\pi \langle H_0 \rangle}{2L} \epc
\end{equation}
where $H_0$ denotes the `kinetic part' of the Hamiltonian,
obtained from  (\ref{hxxz}) by setting $\Delta, h=0$. The
results of a numerical comparison of the left and right hand
side of the sum rule (\ref{sumrule}) are summarized in table
\ref {tab:sum_rule}. We find, in particular, that for the chosen
anisotropies the sum of the $\ell=1$ and $\ell=2$ terms is
already extremely close to the full weight.
\begin{table}[!h]
\begin{center}
\setlength{\tabcolsep}{1em}
\begin{tabular}{cccc}
\toprule
    $\Delta$&         lhs of  (\ref{sumrule}): $\ell=1$&   lhs of  (\ref{sumrule}): $(\ell=1)+(\ell=2)$&        rhs of  (\ref{sumrule})  \\
\midrule[1pt]
          $1.5$&    $1.56692$&     $1.64348$ &    $1.64394$\\
\midrule
          $2$&       $1.36065$&  $1.37615$&     $1.37624$ \\                    
\midrule
          $3$&       $0.987313$&    $0.989092$&     $0.989116$\\                    
\bottomrule
\end{tabular}
\caption{Both sides of the $f$-sum rule (\ref{sumrule}) for $\Delta=1.5, 2$ and 3 and $J=1$.}
\label{tab:sum_rule}
\end{center}
\end{table}

%%%%%%%%%%%%%%%%%%%%%%%%%%%%%%%%%%%%%%%%%%%%%%%%%%
%%%%%%%%%%%%%%%%%%%%%%%%%%%%%%%%%%%%%%%%%%%%%%%%%%
\section{Summary and Conclusions}
%%%%%%%%%%%%%%%%%%%%%%%%%%%%%%%%%%%%%%%%%%%%%%%%%%
%%%%%%%%%%%%%%%%%%%%%%%%%%%%%%%%%%%%%%%%%%%%%%%%%%
We have presented an exact thermal form factor expansion for the
dynamical current-current correlation function $\bigl\< {\cal J}_1 (t)
{\cal J}_{m+1} \bigr\>$ of the spin-1/2 XXZ chain in the massive
antiferromagnetic regime at zero temperature. In this expansion,
the correlation function is given as a sum over $\ell$ particle-hole
excitations or, equivalently, $2\ell$ spinon excitations. The formula
can, in principle, be evaluated numerically for arbitrary distances
$m$ and times $t$, leading to numerically exact results. We note,
in particular, that the series in the number of particle-hole
excitations $\ell$ converges fast, except for anisotropy $\Delta\to 1$.
The long-time, large-distance asymptotics is determined by the $\ell=1$
contribution. We attribute the fast convergence to the massive nature
of the involved excitations.

We have also provided a form factor series representation for
$\lim_{L\to\infty} \frac{1}{L}\sum_{j,m=1}^L \bigl\< {\cal J}_j (t)
{\cal J}_{m} \bigr\>$ which allows to calculate the real
part of the optical spin conductivity ${\rm Re}\, \sigma(\omega)$
by a direct Fourier transform. For the $\ell=1$ (2-spinon)
contribution the Fourier transform can be performed analytically,
leading to a closed form expression for the 2-spinon optical
conductivity. We find that ${\rm Re}\, \sigma(\omega)$ is finite
only within the 2-spinon band which starts at some finite frequency.
At both edges of the spinon band, the conductivity shows a
square-root behavior. By checking the $f$-sum rule, we have
shown that the $\ell=1$ and $\ell=2$ contributions account for
almost the entire spectral weight if we are not too close to the
isotropic point.

Another test of our form factor series for $\bigl\< {\cal J}_1 (t)
{\cal J}_{m+1} \bigr\>$ was provided by DMRG results. We also note that
the obtained ${\rm Re}\,\sigma(\om)$ looks quite similar to the finite
temperature results in Ref.~\cite{KKM14}, which were obtained by DMRG as
well, except for small frequencies. The Lorentzian-type peak around
$\om=0$ observed in this paper, which seems to decrease with increasing
$T$, therefore appears to be a genuine finite-temperature effect related
to the expected diffusive behavior. To understand the low-frequency
behavior better, it would therefore be of interest to extend our form
factor series expansion to finite temperatures.

This is not totally out of reach, since the thermal form factor
approach is a genuine finite-temperature method which only has been
used in the zero-temperature limit here to produce fully explicit
results. One of our future goals is indeed to keep the temperature
finite. For this purpose we will need better control of the
non-linear integral equations that describe the excited states of
the quantum transfer matrix. Simplifications should occur in
the high-temperature limit, where we have a rather complete
understanding \cite{GGKS20} of the involved auxiliary functions.

Further future goals include a proof of the convergence of
the series and an estimation of the truncation error. Given
the explicit nature of the integrands in our series and
the recent progress in related cases \cite{GKS20b,Kozlowski20app}
this may now appear within reach. We shall also work out
thermal form factor series expansions of general two-point
functions of spin zero operators and provide the details of
the proof of (\ref{curcur}), (\ref{lplh}) in a forthcoming
publication.

%%%%%%%%%%%%%%%%%%%%%%%%%%%%%%%%%%%%%%%%%%%%%%%%%%
\section*{Acknowledgment}
%%%%%%%%%%%%%%%%%%%%%%%%%%%%%%%%%%%%%%%%%%%%%%%%%%
The authors would like to thank Z. Bajnok, H. Boos, A. Kl\"umper, F. Smirnov
and A. Wei{\ss}e for helpful discussions. FG and JSi acknowledge
financial support by the German Research Council (DFG) in the framework
of the research unit FOR 2316. KKK is supported by CNRS Grant PICS07877
and by the ERC Project LDRAM: ERC-2019-ADG Project 884584. JSi
acknowledges support by the Natural Sciences and Engineering Research
Council (NSERC, Canada).  JSu is supported by a JSPS KAKENHI Grant
number 18K03452.

\begin{appendix}
%%%%%%%%%%%%%%%%%%%%%%%%%%%%%%%%%%%%%%%%%%%%%%%%%%
%%%%%%%%%%%%%%%%%%%%%%%%%%%%%%%%%%%%%%%%%%%%%%%%%%
\section{Special functions}
\label{app:specialfun}
%%%%%%%%%%%%%%%%%%%%%%%%%%%%%%%%%%%%%%%%%%%%%%%%%%
%%%%%%%%%%%%%%%%%%%%%%%%%%%%%%%%%%%%%%%%%%%%%%%%%%
% \renewcommand{\theequation}{A.\arabic{equation}}
% \setcounter{equation}{0}
In this appendix we gather the definitions of the special
functions needed in the main part of the text and list 
some of their properties.

We start with functions that can be expressed in terms of infinite
$q$-multi factorials which, for $|q_j| < 1$ and $a \in {\mathbb C}$,
are defined as
\begin{equation}
     (a;q_1, \dots, q_p) =
        \prod_{n_1, \dots, n_p = 0}^\infty (1 - a q_1^{n_1} \dots q_p^{n_p}) \epp
\end{equation}
A first set of such functions are the $q$-Gamma and $q$-Barnes
functions $\G_q$ and $G_q$,
\begin{subequations}
\begin{align}
     \G_q (x) & = (1 - q)^{1 - x} \frac{(q;q)}{(q^x;q)} \epc \\[1ex]
     G_q (x) & = (1 - q)^{- \2 (1 - x)(2 - x)} (q;q)^{x - 1}
                 \frac{(q^x;q,q)}{(q;q,q)} \epp
\end{align}
\end{subequations}
They satisfy the normalization conditions
\begin{equation}
     \G_q (1) = G_q (1) = 1
\end{equation}
and the basic functional equations
\begin{equation}
     [x]_q \G_q (x) = \G_q (x + 1) \epc \qd \G_q (x) G_q (x) = G_q (x + 1) \epc
\end{equation}
where $[x]_q = (1 - q^x)/(1 - q)$ is a familiar form of the $q$-number.

Closely related are the Jacobi theta functions $\dh_j (x) =
\dh_j (x|q)$, $j = 1, \dots, 4$. Setting $q = \re^{- \g}$
they can be introduced by
\begin{equation}
     \dh_4 (x|q) = (q^2;q^2) (e^{- 2\i x} q; q^2) (\re^{2 \i x} q; q^2)
\label{definition theta 4}     
\end{equation}
and
\begin{align}
     \dh_1 (x) & = - \i q^\4 \re^{\i x} \dh_4 (x + \i \g/2) \epc \qd
     \dh_2 (x) = q^\4 \re^{\i x} \dh_4 (x + \i \g/2 + \p/2) \epc
        \notag \\[1ex]
     \dh_3 (x) & = \dh_4 (x + \p/2) \epp
\label{definition theta 1 2 3}
\end{align}
The parameter $q$ of the theta functions is called `the nome'.
Sometimes we suppress their nome dependence, but only if the
value of $q$ is clear from the context.

The Jacobi theta functions are connected with the $q$-gamma functions
through the second functional relation of the latter,
\begin{equation}
     \frac{\dh_4 (x|q)}{\dh_4 (0|q)} =
        \frac{\G_{q^2}^2 \bigl(\tst{\2}\bigr)}
             {\G_{q^2} \bigl(\tst{\2 - \frac{\i x}{\g}}\bigr)
              \G_{q^2} \bigl(\tst{\2 + \frac{\i x}{\g}}\bigr)} \epp
\end{equation}
We shall also frequently employ the common notational convention for
the `theta constants', $\dh_j = \dh_j (0|q)$, $j = 2, 3, 4$,
$\dh_1' = \dh_1' (0|q)$.

Another class of functions needed in the main text are the
basic hypergeometric functions \cite{GaRa04}. They are defined
in terms of finite $q$ multi-factorials (or $q$-Pochhammer symbols),
\begin{equation}
    (a_1, \dots, a_k; q)_m = (a_1; q)_m (a_2; q)_m \dots (a_k; q)_m
    \epc \qd (a;q)_m = \prod_{k=0}^{m-1} (1 - a q^k) \epc
\end{equation}
by the infinite series
\begin{equation}
     _r \PH_s \Bigl(\begin{array}{@{}r@{}} a_1, \dots, a_r \\ b_1, \dots, b_s \end{array};
                     q, z\Bigr) =
        \sum_{k=0}^\infty \frac{(a_1, \dots, a_r; q)_k}{(b_1, \dots, b_s, q; q)}
           \Bigl((-1)^k q^\frac{k(k-1)}{2}\Bigr)^{s + 1 - r} z^k \epp
\end{equation}

%%%%%%%%%%%%%%%%%%%%%%%%%%%%%%%%%%%%%%%%%%%%%%%%%%
%%%%%%%%%%%%%%%%%%%%%%%%%%%%%%%%%%%%%%%%%%%%%%%%%%
\section{The spin conductivity of the XXZ chain}
%%%%%%%%%%%%%%%%%%%%%%%%%%%%%%%%%%%%%%%%%%%%%%%%%%
%%%%%%%%%%%%%%%%%%%%%%%%%%%%%%%%%%%%%%%%%%%%%%%%%%
% \renewcommand{\theequation}{B.\arabic{equation}}
% \setcounter{equation}{0}
\label{app:spinconduct}
In this appendix we recall the derivation of several alternative
formulae for the `spin conductivity'.

\subsection{Gauge fields coupling to the Hamiltonian}
We decompose the Hamiltonian (\ref{hxxz}) as
\begin{equation}
     H = H_0 + \D H_I - h S^z \epc
\end{equation}
where
\begin{equation}
     H_0 = 2 J \sum_{j=1}^L \bigl(\s_{j-1}^+ \s_j^- + \s_{j-1}^- \s_j^+\bigr) \epc \qd
     H_I = J \sum_{j=1}^L \bigl(\s_{j-1}^z \s_j^z - 1\bigr) \epp
\end{equation}
Under a Jordan-Wigner transformation the operator $H_0$ goes to
a tight-binding type Hamiltonian, while $H_I$ becomes a nearest-neighbour
density-density interaction. In the Fermion picture, it is $H_0$ which
couples to an external electro-magnetic field via so-called Peierls
phases which can be understood as a manifestation of a $U(1)$ gauge
field. For details see e.g.\ Chapter~1.3 of the book \cite{Thebook}.
In the spin-chain picture, switching on an external field means to
replace
\begin{equation}
     \s_j^- \rightarrow \re^{\i \ph_j (t)} \s_j^- \epc \qd
     \s_j^+ \rightarrow \re^{- \i \ph_j (t)} \s_j^+ \epc
\end{equation}
where $t$ is the time variable. We shall restrict ourselves
to a spatially homogeneous field (`the case of long wave length'),
\begin{equation}
     \ph_j (t) - \ph_{j-1} (t) = \la (t) \epp
\end{equation}
Then $H_0$ turns into
\begin{equation}
     H_\la = 2 J \sum_{j=1}^L
        \bigl(\re^{\i \la(t)} \s_{j-1}^+ \s_j^- +
	      \re^{- \i \la(t)} \s_{j-1}^- \s_j^+\bigr) \epp
\end{equation}

By analogy with the electro-magnetic case we shall assume that the
gauge field is related to the `electric field' $E$ as
\begin{equation}
     \6_t \la (t) = - e a E(t) \epc
\end{equation}
where $e$ is a unit charge and $a$ a unit length (`lattice spacing').
This implies the relation
\begin{equation} \label{lafef}
     \la_F (\om) = - \frac{\i e a}{\om} E_F (\om) \qd \text{with} \qd 
     \la_F (\om) = \int\limits_{\mathbb{R}}^{} \rd t \, \re^{\i \om t} \la(t)
\end{equation}
for the Fourier transforms. We shall consider a class of fields
$\la$ for which $|\la (t)| \le \re^{\e t}$ for $t \rightarrow
- \infty$ and $|\la (t)| \le c t$ for $t \rightarrow \infty$,
where $\e, c > 0$. The first condition is compatible with an
adiabatic switching on of the field and the second one admits 
`electric fields' which are asymptotically constant and allow 
us to probe the dc conductivity. For such fields the Fourier
transform $\la_F (\om)$ exists within a strip $0 < \Im \om < \e$
and should be interpreted as a `$+$-boundary value' on the real axis.
\subsection{Current operators}
An external `electric field' will induce a current into a wire.
Let us recall the construction of the corresponding current
operator.

We start with the definition of the operator of the time
derivative of a physical quantity in the Schr\"odinger picture.
The Schr\"odinger equation,
\begin{equation} \label{timeevolution}
     \i \6_t U(t) = H(t) U(t) \epc \qd U(0) = \id \epc
\end{equation}
determines the time evolution operator $U(t)$ for a system
with generally time dependent Hamiltonian $H(t)$. If $A$ is
any operator in the Schr\"odinger picture, then the corresponding
operator $A_H$ in the Heisenberg picture is
\begin{equation} \label{heispic}
     A_H (t) = U^{-1} (t) A U(t) \epp
\end{equation}
Equations (\ref{timeevolution}) and (\ref{heispic}) imply that
\begin{equation}
     \i \6_t A_H (t) = - U^{-1} (t) [H(t), A] U(t) 
\end{equation}
or
\begin{equation}
     U(t) (\6_t A_H (t)) U^{-1} (t) = \i [H(t), A] = \dot A \epp
\end{equation}
We may think of this equation as defining the time derivative
$\dot A$ of $A$ in the Schr\"odinger picture.

Applying this to the local magnetization $\2 \s^z$ and $H(t)
= H_\la + \D H_I - h S^z$ we obtain
\begin{equation} \label{szconti}
     \tst{\2} \dot \s_j^z = - J_{j+1} (t) + J_j (t) \epc
\end{equation}
where
\begin{equation}
     J_j (t) = 2 \i J 
        \bigl(\re^{\i \la(t)} \s_{j-1}^+ \s_j^- -
	      \re^{- \i \la(t)} \s_{j-1}^- \s_j^+\bigr) \epp
\end{equation}
Equation (\ref{szconti}) has the form of a continuity equation for
the local magnetization. For this reason $J_j (t)$ is interpreted
as the density of the spin current.

Let
\begin{equation}
     {\cal J} = \sum_{m=1}^L J_m (0) = \sum_{m=1}^L {\cal J}_m \epp
\end{equation}
Then the total magnetic current is the sum
\begin{equation} \label{jtlinlambda}
     J(t) = \sum_{m=1}^L J_m (t) = {\cal J} - \la(t) H_0
            + {\cal O} (\la^2)
\end{equation}
and the time dependent Hamiltonian has the expansion
\begin{equation} \label{htlinlambda}
     H_\la + \D H_I - h S^z = H + \la(t) {\cal J} + {\cal O} (\la^2) \epp
\end{equation}
The latter two equations are all we need in order to calculate
the average current induced by the external field within the linear
response theory. The small time dependent perturbation we can
read off from (\ref{htlinlambda}) is $V(t) = \la(t) {\cal J}$.

\subsection{Linear response of the current}
We denote the density matrix of the canonical ensemble by $\r_c$
and the density matrix obtained by time evolving $\r_c$ with
$H_\la$ by $\r (t)$. Then the linear response formula
\begin{equation}
     \tr \bigl\{\bigl(\r(t) - \r_c\bigr) J(t)\bigr\}
        = - \i \int_{- \infty}^t \rd t' \:
	       \bigl\<[(J(t))_H (t - t'), V(t')]\bigr\>_T
\end{equation}
determines the averaged current to linear order in $V$ (for a
concise derivation see e.g.\ Section~L.22 of \cite{Goehmann21}).
Inserting here (\ref{jtlinlambda}) and (\ref{htlinlambda}) we
obtain
\begin{equation} \label{jtlinresponse}
     \tr \bigl\{\r(t) J(t)\bigr\} =
	- \<H_0\>_T \la (t)
        - \i \int_{- \infty}^\infty \rd t' \: \Th (t - t')
	     \bigl\<[{\cal J} (t - t'), {\cal J}]\bigr\>_T \la (t')
	     + {\cal O} (\la^2) \epp
\end{equation}
In this equation $\Th$ is the Heaviside step function and
${\cal J} (t)$ denotes the total current ${\cal J}$ in the
Heisenberg picture that is evolved with respect to the
unperturbed Hamiltonian $H$. We have made use of the fact that
$\<{\cal J}\>_T = 0$ due to the invariance of the XXZ
Hamiltonian under parity transformations.

The `experimentally relevant quantity' is the Fourier transformed
current per volume which in physical units is given by
\begin{equation} \label{defjpervol}
     {\cal J}_F (\om) =
        - ea \int_{- \infty}^\infty \rd t \: \re^{\i \om t}
	   \frac{\tr \bigl\{\r(t) J(t)\bigr\}}{a^3 L} \epp
\end{equation}
Due to the remark below (\ref{lafef}), the integral
on the right hand side is to be interpreted as a $+$-boundary
value if $\om$ is real.

If we substitute (\ref{jtlinresponse}) into (\ref{defjpervol}), 
use the convolution theorem as well as (\ref{lafef}) and neglect
all terms of quadratic oder in $\la$ or higher, we arrive at
`Ohm's law',
\begin{equation}
     {\cal J}_F (\om) = \frac{e^2}{a} \s_L (\om) E_F (\om) \epc
\end{equation}
where $\s_L (\om)$ is the specific optical conductivity,
\begin{equation} \label{sidiss}
     \s_L (\om) = \frac{1}{L (\om + \i 0)} \biggl\{- \i \<H_0\>_T
        + \int_0^\infty \rd t \: \re^{\i \om t}
	   \bigl<[{\cal J} (t), {\cal J}]\bigr\>_T \biggr\} \epp 
\end{equation}
Assuming $\bigl<[{\cal J} (t), {\cal J}]\bigr\>_T$ to be bounded
for $t \rightarrow + \infty$ we see that the right hand side
of (\ref{sidiss}) is a holomorphic function of $\om$ in the
upper half plane. This implies that the real part and the imaginary
part of the optical conductivity are not independent, but are
connected by the Kramers-Kronig relation.

\subsection{Real part of the optical conductivity}
For this reason we can focus our attention on the real part of
the conductivity. We wish to rewrite it in a form appropriate
for taking the thermodynamic limit. We basically follow the
arguments given in \cite{Sirker20} and start by switching
to a spectral representation of the integral on the right hand
side of (\ref{sidiss}). Employing the notation
\begin{equation}
     Z_\la = \tr \Bigl\{\re^{- \frac{1}{T}(H_\la + \D H_I - h S^z)}\Bigr\} \epc \qd
     p_n = \frac{\re^{- \frac{E_n}{T}}}{Z_0} \epc \qd
     \om_{mn} = E_m - E_n \epc
\end{equation}
where the $E_n$ are the eigenvalues of the Hamiltonian (\ref{hxxz})
with corresponding eigenstates $|n\>$, the spectral representation takes the form
\begin{equation}
     \int_0^\infty \rd t \: \re^{\i \om t} \bigl<[{\cal J} (t), {\cal J}]\bigr\>_T
        = \i \sum_{\substack{m, n\\E_m \ne E_n}}
	     \frac{p_n - p_m}{\om - \om_{mn} + \i 0}
	     \bigl|\<m|{\cal J}|n\>\bigr|^2 \epp
\end{equation}
Now, if $E_m \ne E_n$,
\begin{equation}
     \frac{1}{\om + \i 0} \cdot \frac{1}{\om - \om_{mn} + \i 0}
        = \frac{1}{\om_{mn}}
	  \biggl(\frac{1}{\om - \om_{mn} + \i 0}
	         - \frac{1}{\om + \i 0}\biggr) \epp
\end{equation}
Using this identity as well as the Plemelj formula $1/(\om + \i 0)
= - \i \p \de(\om) + {\cal P}(1/\om)$ we obtain the spectral
representation
\begin{multline} \label{reomspectral}
     \Re \s_L (\om) = \frac{\p}{L} \biggl\{- \<H_0\>_T
        + \sum_{\substack{m, n\\E_m \ne E_n}}
	     \frac{p_m - p_n}{\om_{mn}}
	     \bigl|\<m|{\cal J}|n\>\bigr|^2 \biggr\} \de(\om) \\
        - \frac{\p}{L} \sum_{\substack{m, n\\E_m \ne E_n}}
	     \frac{p_m - p_n}{\om_{mn}}
	     \bigl|\<m|{\cal J}|n\>\bigr|^2 \de(\om - \om_{mn}) \epp
\end{multline}
This representation immediately implies the f-sum rule (\ref{sumrule}).

Now notice that the free energy per lattice site,
\begin{equation}
     f(\la) = - \frac{T}{L} \ln Z_\la \epc
\end{equation}
satisfies \cite{GiamarchiShastri} the relation
\begin{equation} \label{meissnerfrac}
     \6_\la^2 f(\la)\bigr|_{\la = 0} = - \frac{\<H_0\>_T}{L}
        + \frac{1}{L} \sum_{\substack{m, n\\E_m \ne E_n}}
	   \frac{p_m - p_n}{\om_{mn}} \bigl|\<m|{\cal J}|n\>\bigr|^2
        - \frac{1}{TL} \sum_{\substack{m, n\\E_m = E_n}}
	                  p_m \bigl|\<m|{\cal J}|n\>\bigr|^2 \epp
\end{equation}
This quantity is the so-called Meissner fraction. It vanishes in the
thermodynamic limit \cite{GiamarchiShastri}, which follows from the
fact that the effect of the external field $\la$ is equivalent to
a mere twist of the boundary conditions of the original Hamiltonian
(\ref{hxxz}). Inserting (\ref{meissnerfrac}) into (\ref{reomspectral})
and switching back from a spectral representation to a Fourier
integral we obtain
\begin{multline}
     \Re \s_L (\om) = \p \de(\om) \6_\la^2 f(\la)\bigr|_{\la = 0} +
        \frac{1 - \re^{- \frac{\om}{T}}}{2 \om L}
	   \int_{- \infty}^\infty \rd t \: \re^{\i \om t}
	     \bigl<{\cal J} (t) {\cal J}\bigr\>_T \\ =
        \p \de(\om) \6_\la^2 f(\la)\bigr|_{\la = 0} +
        \frac{\re^{\frac{\om}{T}} - 1}{2 \om L}
	   \int_{- \infty}^\infty \rd t \: \re^{- \i \om t}
	     \bigl<{\cal J} (t) {\cal J}\bigr\>_T \epp
\end{multline}
From here it is obvious that $\Re \s_L (\om)$ is even. Since the
Meissner fraction vanishes in the thermodynamic limit, we obtain
the formula
\begin{equation} \label{resigfluctform}
     \Re \s (\om) = \lim_{L \rightarrow \infty} \Re \s_L (\om) =
        \frac{1 - \re^{- \frac{\om}{T}}}{2 \om}
	   \int_{- \infty}^\infty \rd t \: \re^{\i \om t}
	     \lim_{L \rightarrow \infty}
	     \frac{\bigl<{\cal J} (t) {\cal J}\bigr\>_T}{L}
\end{equation}
that is used in the main text.

%%%%%%%%%%%%%%%%%%%%%%%%%%%%%%%%%%%%%%%%%%%%%%%%%%
%%%%%%%%%%%%%%%%%%%%%%%%%%%%%%%%%%%%%%%%%%%%%%%%%%
\section{Two-spinon contribution}
%%%%%%%%%%%%%%%%%%%%%%%%%%%%%%%%%%%%%%%%%%%%%%%%%%
%%%%%%%%%%%%%%%%%%%%%%%%%%%%%%%%%%%%%%%%%%%%%%%%%%
% \renewcommand{\theequation}{C.\arabic{equation}}
% \setcounter{equation}{0}
\label{app:twospincond}
\subsection{Two-spinon dynamical structure function}
Defining
\begin{equation} \label{defsj}
     S_{\cal J}^{(2 \ell)} (Q, \om) =
        \sum_{m = - \infty}^\infty \int_{- \infty}^\infty \rd t \:
        \re^{\i(Q m + \om t)} C^{(2 \ell)} (m, t) 
\end{equation}
the function
\begin{equation}
     S_{\cal J} (Q, \om) = \sum_{\ell = 1}^\infty S_{\cal J}^{(2 \ell)} (Q, \om)
        % =
        %\sum_{m = - \infty}^\infty \int_{- \infty}^\infty \rd t \:
        %\re^{\i(Q m + \om t)}
	%\bigl\< {\cal J}_1 (t) {\cal J}_{m+1} \bigr\>
\end{equation}
is called the dynamical structure function of the local magnetic
currents. In the following we shall obtain an explicit expression
for the one-particle one-hole term $S_{\cal J}^{(2)} (0, \om)$.

For this purpose we start with a close inspection and simplification
of the amplitude
\begin{equation} \label{amp1}
     {\cal A}_{\cal J}^{(2)} (u,v|k) = 
        \biggl(\frac{(\e(u) - \e(v)) \dh_1'}
	            {4 \sin(u - v) \dh_1 (\Si)}\biggr)^2
        \frac{\Xi^2 (0) {\cal M}\, \hat {\cal M}}
	     {\Xi(u - v) \Xi(v - u)} \epc
\end{equation}
where
\begin{subequations}
\begin{align}
     \Si & = - \2 (u - v + \p k) \epc \qd
     \Xi (z) = \frac{\G_{q^4} \bigl(\tst{\2 + \frac{z}{2 \i \g}}\bigr)}
                      {\G_{q^4} \bigl(\tst{1 + \frac{z}{2 \i \g}}\bigr)}
                 \frac{G_{q^4}^2 \bigl(\tst{1 + \frac{z}{2 \i \g}}\bigr)}
                      {G_{q^4}^2 \bigl(\tst{\2 + \frac{z}{2 \i \g}}\bigr)} \epc \\[.5ex]
     \label{phnull}
     {\cal M} & = \overline{\PH}_1 (P; 0) - \PH_1 (P; 0)
                                          \frac{\Ph^{(-)} (v)}{\Ph^{(+)} (v)} \epc \qd
     \hat {\cal M} = \hat \PH_1 (H; 0) - \hat{\overline{\PH}}_1 (H; 0)
                                         \frac{\Ph^{(-)} (u)}{\Ph^{(+)} (u)}
\end{align}
\end{subequations}
with $H = \re^{2 \i u}, P = \re^{2 \i v}$ and
\begin{subequations}
\begin{align}
     \frac{\Ph^{(-)} (v)}{\Ph^{(+)} (v)} & = \frac{\Ph^{(-)} (u)}{\Ph^{(+)} (u)} =
        \re^{- 2 \i \Si}
	\frac{\G_{q^4} \bigl(\tst{\2 + \frac{\i (v - u)}{2 \g}}\bigr)
	      \G_{q^4} \bigl(\tst{1 - \frac{\i (v - u)}{2 \g}}\bigr)}
	     {\G_{q^4} \bigl(\tst{\2 - \frac{\i (v - u)}{2 \g}}\bigr)
	      \G_{q^4} \bigl(\tst{1 + \frac{\i (v - u)}{2 \g}}\bigr)} \epc \\[1ex]
     \PH_1 (P; 0) & = \hat{\overline{\PH}}_1 (H; 0) =
        \phantom{}_2 \PH_1
        \biggl(\begin{array}{@{}r}
	          q^{-2}, P/H \\ q^2 P/H
	       \end{array}
	       ;q^4, q^4 \biggr) \epc \\
     \overline{\PH}_1 (P; 0) & = \hat \PH_1 (H; 0) =
        \phantom{}_2 \PH_1
        \biggl(\begin{array}{@{}r}
	          q^{-2}, H/P \\ q^2 H/P
	       \end{array}
	       ;q^4, q^4 \biggr) \epp
\end{align}
\end{subequations}
Using the $q$-Gau{\ss} identity \cite{GaRa04},
\begin{equation} \label{qgauss}
     \phantom{}_2 \PH_1
        \biggl(\begin{array}{@{}r}
	          q^{-2}, H/P \\ q^2 H/P
	       \end{array}
	       ;q^4, q^4 \biggr) =
     \frac{\G_{q^4} \bigl(\tst{\2 + \frac{\i (v - u)}{2 \g}}\bigr)}
          {\G_{q^4} (\tst{\2}) \G_{q^4} \bigl(\tst{1 + \frac{\i (v - u)}{2 \g}}\bigr)} \epc
\end{equation}
as well as the functional equations for the $q$-gamma and $q$-Barnes
functions, the amplitude can be rewritten as
\begin{multline} \label{amp12}
     {\cal A}_{\cal J}^{(2)} (u, v + \i \g|k) =
        \biggl(\frac{\e(u) + \e(v)}{2}\biggr)^2
        \frac{(-1)^k q^\2 \tg \bigl(\2 (u - v - \i \g + \p k)\bigr)}
	       {2 \sin(u - v)} \\[-.5ex]
	\times B(u - v)
        \biggl(\frac{\dh_1'}{\dh_4 \bigl(\2 (u - v + \p k)\bigr)}\biggr)^2 \epc
\end{multline}
where $B(z)$ was defined in equation (\ref{randb}) of the main text.
Note that $B(z)$ has a double zero at $z = 0$. Hence, the simple
pole at $u = v$ stemming from the sine function is canceled by a
double zero of $B(u - v)$.

For this reason we can write
\begin{equation} \label{ctwosym}
     C^{(2)} (m, t) =
        \sum_{k = 0, 1}
	   \int_{- \frac{\p}{2}}^\frac{\p}{2} \frac{\rd z_1}{2\p}
	   \int_{- \frac{\p}{2}}^\frac{\p}{2} \frac{\rd z_2}{2\p} \:
	   {\cal A}^{(2)}_s (z_1, z_2|k)
	   \re^{\i m k \p + \i \sum_{j=1}^2 (t \e (z_j) - m p(z_j))} \epc
\end{equation}
where
\begin{multline} \label{atwosym}
     {\cal A}^{(2)}_s (z_1, z_2|k) = \2
        \bigl( {\cal A}_{\cal J}^{(2)} (z_1, z_2 + \i \g|k) +
               {\cal A}_{\cal J}^{(2)} (z_2, z_1 + \i \g|k)\bigr) \\[.5ex] =
	\frac{q^\2}{2} \biggl(\frac{\e(z_1) + \e(z_2)}{2}\biggr)^2
        \frac{B(z_1 - z_2)} {\D + (-1)^k \cos(z_1 - z_2)}
        \frac{(\dh_1')^2}{\dh_4^2 \bigl(\2 (z_1 - z_2 + \p k)\bigr)} \epp
\end{multline}
It follows that
\begin{multline}
     S_{\cal J}^{(2)} (Q, \om) = 
        \sum_{k = 0, 1}
	   \int_{- \frac{\p}{2}}^\frac{\p}{2} \rd z_1
	   \int_{- \frac{\p}{2}}^\frac{\p}{2} \rd z_2 \:
           {\cal A}^{(2)}_s (z_1, z_2|k) \\ \times
	   \de_{2\p} \bigl(Q - p(z_1) - p(z_2) + \p k\bigr)
	   \de \bigl(\om + \e(z_1) + \e(z_2)\bigr) \epc
\end{multline}
where $\de_{2 \p}$ is a $2 \p$-periodic delta function.

We now substitute
\begin{equation}
     \begin{pmatrix} z_1 \\ z_2 \end{pmatrix} \mapsto
        \begin{pmatrix} \la \\ P \end{pmatrix} =
	\begin{pmatrix} \2 \bigl(p(z_1) - p(z_2)\bigr)
	                \\ p(z_1) + p(z_2) \end{pmatrix} \epp
\end{equation}
For the substitution recall \cite{DGKS15a} that $p(x)$ is
monotonically increasing on $[-\p/2,\p/2]$ with $p(-\p/2) = 0$,
$p(\p/2) = \p$. Furthermore, the inverse function can be
written as
\begin{equation}
     p^{-1} (y) = - \frac{\p}{2 K} \arcsn \bigl(\cos(y)\big|k \bigr) \epc
\end{equation}
where $k$ is the elliptic module and $K$ the complete elliptic
integral (see (\ref{ellipticcons})). Setting
\begin{equation}
     A(P, \la|k) =
        \frac{{\cal A}^{(2)}_s \bigl(p^{-1} (P/2 + \la),
	       p^{-1} (P/2 - \la)\big|k\bigr)}
	     {p' \bigl(p^{-1} (P/2 + \la)\bigr)
	      p' \bigl(p^{-1} (P/2 - \la)\bigr)}
\end{equation}
we obtain
\begin{multline}
     S_{\cal J}^{(2)} (Q, \om) = 
        \sum_{k = 0, 1}
	   \biggl\{ \int_0^\p \rd P \int_{- \frac{P}{2}}^\frac{P}{2} \rd \la
	            + \int_{  \p}^{2 \p} \rd P
		      \int_{  -\p + \frac{P}{2}}^{\p - \frac{P}{2}} \rd \la
		      \biggr\} A (P, \la|k) \\[1ex] \times
	   \de_{2\p} \bigl(Q - P + \p k\bigr)
	   \de \bigl(\om + \e(p^{-1} (P/2 + \la)) + \e(p^{-1}(P/2 - \la))\bigr) \epp
\end{multline}
In the latter equation the $P$ integration is now trivial.
For $Q \in (0, \p)$ we obtain
\begin{multline} \label{sqqlesspi}
     S_{\cal J}^{(2)} (Q, \om) = 
        \int_{- \frac{Q}{2}}^\frac{Q}{2} \rd \la \: A (Q, \la|0)
	   \de \bigl(\om + \e(p^{-1} (Q/2 + \la)) + \e(p^{-1}(Q/2 - \la))\bigr) \\ +
        \int_{- \frac{\p - Q}{2}}^\frac{\p - Q}{2} \rd \la \: A (Q + \p, \la|1)
	   \de \bigl(\om + \e(p^{-1} (\tst{\frac{ Q + \p}{2} + \la}))
	                 + \e(p^{-1}(\tst{\frac{Q + \p}{2} - \la}))\bigr) \epp
\end{multline}

In the limit $Q \rightarrow 0$ the first integral can at most
contribute to the value of $S_{\cal J}^{(2)}$ at the single point $(0,-2\varepsilon(-\pi/2))$. We shall
ignore this irregular contribution. Taking into account that
$A (P, \la|k) = A(P, - \la|k)$ we see that at all other points
\begin{equation}
     S_{\cal J}^{(2)} (0, \om) = 
        2 \int_0^\frac{\p}{2} \rd \la \: A (\p, \la|1)
	   \de \bigl(\om + \e(p^{-1} (\tst{\frac{\p}{2} + \la}))
	                 + \e(p^{-1}(\tst{\frac{\p}{2} - \la}))\bigr) \epp
\end{equation}
Further noticing that
\begin{equation}
     \e(p^{-1} (\tst{\frac{\p}{2} \pm \la})) =
        - \frac{h_\ell}{2 k'} \sqrt{1 - k^2 \sin^2 (\la)} \quad \text{and} \quad 
\e(\la)=-2J\sinh( \gamma) p^{\prime}(\la) \;, 
\end{equation}
see (A.11) of \cite{DGKS15b} and (A.19) of \cite{DGKS15a} for
more details,\footnote{One should incorporate $2\pi$ into
$p^{\prime}$ appearing in these works due to the different 
conventions we use here for the dressed momentum.} we can
readily calculate the remaining integral. Using (\ref{randb})
we arrive at
\begin{equation}
\label{SJ}
     S_{\cal J}^{(2)} (0, \om) = \frac{q^\2 h_\ell^2 k}{4 k'}
        \frac{B\bigl(r (\om)\bigr)}{\D - \cos\bigl(r(\om)\bigr)}
	\frac{\dh_3^2}{\dh_3^2 \bigl(r(\om)/2\bigr)}
	\frac{\om}{\sqrt{\bigl((h_\ell/k')^2 - \om^2\bigr)
	                 \bigl(\om^2 - h_\ell^2\bigr)}}
\end{equation}
for $\om \in [h_\ell, h_\ell/k']$. Outside this interval the function
$S_{\cal J}^{(2)} (0, \om)$ vanishes.

 The first integral on the right hand side of (\ref{sqqlesspi})
can at most contribute to $S_{\cal J}^{(2)} (0,\om)$ at $\om = - 2 \, \e (p^{-1} (0))
= h_\ell$, which means exactly at the lower band edge.

\subsection{Two-spinon optical conductivity}
We would like to connect the two-spinon contribution to the
structure function with the real part of the optical conductivity.
For this purpose we first note that
\begin{equation} \label{jjconju}
     \bigl\< {\cal J}_1 (t) {\cal J}_{m+1} \bigr\> =
        \Big( \bigl\< {\cal J}_1 (- t) {\cal J}_{m+1} \bigr\> \Big)^* \epp
\end{equation}
In order to see this we start with a finite chain of length $L$
for which
\begin{multline} \label{jjconjuproof}
     \Big( \bigl\< {\cal J}_1 (- t) {\cal J}_{m+1} \bigr\> \Big)^* =
        \bigl\< ({\cal J}_1 (- t) {\cal J}_{m+1})^\dagger \bigr\> =
        \bigl\< {\cal J}_{m + 1} {\cal J}_1 (-t) \bigr\> \\[1ex] =
        \bigl\< {\cal J}_{m + 1} (t) {\cal J}_1 \bigr\> =
        \bigl\< {\cal J}_L (t) {\cal J}_{L - m} \bigr\> =
        \bigl\< {\cal J}_1 (t) {\cal J}_{m + 1} \bigr\> \epp
\end{multline}
Here we have used the invariance under parity transformations
in the last equation. Equation (\ref{jjconjuproof}) holds for
every finite $L$, hence also in the thermodynamic limit.

Now (\ref{jjconju}) implies
\begin{equation}
     \bigl\< {\cal J}_1 (t) {\cal J}_{m+1} \bigr\> =
        \sum_{\ell = 1}^\infty \Big( C^{(2 \ell)} (m, - t) \Big)^* \epp
\end{equation}
Thus, the two-spinon contribution to the correlation function
of the total currents per lattice site is
\begin{multline}
     \biggl(2 \sum_{m = 0}^\infty \bigl\< {\cal J}_1 (t) {\cal J}_{m+1} \bigr\>_T
            - \bigl\< {\cal J}_1 (t) {\cal J}_1 \bigr\>_T \biggr)^{(2)} \\ =
	    \sum_{m = 0}^\infty C^{(2)} (m, t) +
	    \sum_{m = 1}^\infty  \Big( C^{(2)} (m, - t) \Big)^*
	    = \sum_{m = - \infty}^\infty C^{(2)} (m, t) \epc
\end{multline}
as can be seen from taking the complex conjugation of the explicit
expression \eqref{ctwosym}. Hence, with (\ref{redyncond}) and (\ref{defsj}),
\begin{equation}
     \Re \s^{(2)} (\om) = \frac{S_{\cal J}^{(2)} (0,\om)}{2 \om}
\end{equation}
which is valid for all $\om > 0$ and $T = 0$.
Lemma~\ref{lem:twoconduct} and Eq.~\eqref{sigmaTwospinon} in the main text therefore follow directly from Eq.~\eqref{SJ}.

\subsection{The isotropic limit: $\ell=1$}
We consider ${\rm Re}\, \sigma^{(2)}(\om)$  near the lower 2-spinon
band edge in the isotropic limit $\gamma \rightarrow 0$.
A convenient parameter in this limit is $q'={\rm e}^{-\frac{\pi^2}{2\gamma}}$
which approaches zero quickly. The energy gap $\Delta \varepsilon$
in the absence of a magnetic field is
\[
     \Delta \varepsilon = - \varepsilon(\frac{\pi}{2})
        \sim  \frac{8 J \pi \sh \gamma}{\gamma} q' \epp
\]
Our main goal is to show that the peak location $\om^*$
of ${\rm Re}\, \sigma^{(2)}$ is parameterized as $\om^* \sim C q'$
while the corresponding height is  given by ${\rm Re}\, \sigma^{(2)}(\om^*) \sim 
C^{\prime} /q'$ for some constants $C, C^{\prime}$.

The various constants behave in this limit as follows,
\begin{align*}
     k & \sim 1 & k' & \sim 4  q' \\
     K(k) & \sim \frac{\pi^2}{2\gamma} &
     h_{\ell} & \sim \frac{16 J \pi \sh \gamma}{\gamma} q' \epp
\end{align*}
The upper 2-spinon band edge $\frac{h_{\ell}}{k'}$ thus
reaches $4 J \pi $. %\frac{4 J \pi \sh \gamma}{\gamma}$.

Note that  $r(\om)=\pi$ at the lower 2-spinon band edge,
$\om=h_{\ell}$.  Then we conveniently parameterize
\begin{align}
     r(\om) & = \pi - \frac{\gamma}{\pi} \epsilon_r &
     \om &=h_{\ell}(1+\epsilon') \epp
\end{align}
We assume that $\epsilon_r, \epsilon'$ are  $O(1)$.
They are not independent but constrained by (\ref{eq:def_r}),
\[
     \sn(K + \frac{K}{\pi^2} \gamma \epsilon_r|k)
        = {\rm cd} (\frac{K}{\pi^2}  \gamma  \epsilon_r|k)
	= \frac{\sqrt{(\frac{h_{\ell}}{k'})^2 - \om^2}}{\frac{h_{\ell}k}{k'}} \epp
\]

%The lhs is expanded as
%\[
%1+ 4 q' (1-\ch \epsilon_r) +O((q')^2)
%\]
%while the lhs is
%\[
%1- (\frac{k'}{k})^2 ( \epsilon' + \frac{(\epsilon')^2}{2}) + O(  (\frac{k'}{k})^4)
%= 1- 16 q' ( \epsilon' + \frac{(\epsilon')^2}{2}) + O((q')^2) \epp
%\]
%Thus $\epsilon_r$ and $\epsilon'$  are related by
By expanding both sides up to $O((q')^2)$, we find that $\epsilon_r$
and $\epsilon'$  are related by
\[
     \ch \epsilon_r -1 =4( \epsilon' + \frac{(\epsilon')^2}{2}) \epp
\]
In particular, when both of them are infinitesimally small, $\epsilon_r
\sim 2\sqrt{2\epsilon'}$. This essentially explains the
$\sqrt{\om^2-h_{\ell}^2}$ behaviour of ${\Re}\, \sigma^{(2)}(\om)$
for generic $\gamma$.

Now that $B(\pi+z)=B(z)$, we have in this limit, 
\[
     B(r(\om)) = B(\frac{\gamma \epsilon_r}{\pi})
        \sim \frac{1}{G^4(\frac{1}{2})}
        \frac{\epsilon_r}{2\pi^2} \sh \frac{\epsilon_r}{2}
        \prod_{\sigma=\pm1} \frac{G^2(1 + \frac{\sigma\epsilon_r}{2i\pi})}
	                         {G(\frac{1}{2}+ \frac{ \sigma\epsilon_r}{2i\pi})
				  G(\frac{3}{2}+   \frac{ \sigma\epsilon_r}{2i\pi})} \epc
\]
where $G$ is the (undeformed)  Barnes $G$ function.
The limits of the other factors in (\ref{sigmaTwospinon}) are
easily expressed in terms of $\epsilon_r$,
\begin{align*}
     \frac{\dh_3^2}{\dh_3^2 \bigl(r(\om)/2\bigr)} &
     \sim \frac{1}{4 q' \ch^2 \frac{\epsilon_r}{2}} \\
     \frac{1}{\sqrt{\bigl((h_\ell/k')^2 - \om^2\bigr)
              \bigl(\om^2 - h_\ell^2\bigr)}} &
	      \sim  \frac{k'}{h^2_{\ell}k \sh \frac{\epsilon_r}{2}} \\
     \frac{1}{\Delta -\cos(r(\om))} & \sim \frac{1}{2} \epp
\end{align*}

All in all, ${\rm Re}\, \sigma^{(2)}(\om)$ behaves near $h_{\ell}$
in the rational limit as
\begin{align*}
     {\rm Re}\, \sigma^{(2)}(\om)  \sim
        \frac{1}{128 q' G^4(\frac{1}{2}) \pi^2} {\cal F}(\epsilon_r) \epc \quad
	{\cal F}(x) =\frac{x}{\ch^2\frac{x}{2}}
	\prod_{\sigma=\pm1} \frac{G^2(1 + \frac{ \sigma x}{2i\pi})}
	                         {G(\frac{1}{2} + \frac{ \sigma x}{2i\pi})
				  G(\frac{3}{2} + \frac{ \sigma x}{2i\pi})} \epp
\end{align*}
Numerically, ${\cal F}(\epsilon_r)$ has a maximum at $\epsilon_r\sim 1.3508$
and the corresponding peak location is $\om^* \sim 1.2369 h_{\ell} $.
Therefore we conclude that  ${\rm Re}\, \sigma^{(2)}(\om^*)$ behaves as
$1/q'$, while $\om^*$ behaves as $q'$. The above argument only takes into
account the contribution from the $\ell=1$ sector but we expect that the higher
excitations do not alter the qualitative behavior in this limit.

\end{appendix}

%\bibliography{hub,Literatur}

\begin{thebibliography}{10}
\providecommand{\url}[1]{\texttt{#1}}
\providecommand{\urlprefix}{URL }
\expandafter\ifx\csname urlstyle\endcsname\relax
  \providecommand{\doi}[1]{doi:\discretionary{}{}{}#1}\else
  \providecommand{\doi}{doi:\discretionary{}{}{}\begingroup
  \urlstyle{rm}\Url}\fi
\providecommand{\eprint}[2][]{\url{#2}}

\bibitem{BHKPSZ21}
B.~Bertini, F.~Heidrich-Meisner, C.~Karrasch, T.~Prosen, R.~Steinigeweg and
  M.~\v{Z}nidari\v{c},
\newblock \emph{Finite-temperature transport in one-dimensional quantum lattice
  models},
\newblock Rev. Mod. Phys. \textbf{93}, 025003 (2021),
\newblock \doi{10.1103/RevModPhys.93.025003}.

\bibitem{NDMP22}
J.~{De~Nardis}, B.~Doyon, M.~Medenjak and M.~Panfil,
\newblock \emph{Correlation functions and transport coefficients in generalised
  hydrodynamics},
\newblock J. Stat. Mech.: Theor. Exp. p. P014002 ({2022}),
\newblock \doi{{10.1088/1742-5468/ac3658}}.

\bibitem{SPA11}
J.~Sirker, R.~G. Pereira and I.~Affleck,
\newblock \emph{Conservation laws, integrability, and transport in
  one-dimensional quantum systems},
\newblock Phys. Rev. B \textbf{83}, 035115 (2011),
\newblock \doi{10.1103/PhysRevB.83.035115}.

\bibitem{Sirker20}
J.~Sirker,
\newblock \emph{{Transport in one-dimensional integrable quantum systems}},
\newblock SciPost Phys. Lect. Notes p.~17 (2020),
\newblock \doi{10.21468/SciPostPhysLectNotes.17}.

\bibitem{Hlubeketal10}
N.~Hlubek, P.~Ribeiro, R.~Saint-Martin, A.~Revcolevschi, G.~Roth, G.~Behr,
  B.~B\"uchner and C.~Hess,
\newblock \emph{Ballistic heat transport of quantum spin excitations as seen in
  {$SrCuO_2$}},
\newblock Phys. Rev. B \textbf{81}, 020405(R) (2010),
\newblock \doi{10.1103/PhysRevB.81.020405}.

\bibitem{FKECSHBSGGBK13}
T.~Fukuhara, A.~Kantian, M.~Endres, M.~Cheneau, P.~Schau{\ss}, S.~Hild,
  D.~Bellem, U.~Schollw\"ock, T.~Giamarchi, C.~Gross, I.~Bloch and S.~Kuhr,
\newblock \emph{Quantum dynamics of a single, mobile spin impurity},
\newblock Nat. Phys. \textbf{9}, 235 (2013),
\newblock \doi{10.1038/nphys2561}.

\bibitem{FSEHCBG13}
T.~Fukuhara, P.~Schau{\ss}, M.~Endres, S.~Hild, M.~Cheneau, I.~Bloch and
  C.~Gross,
\newblock \emph{Microscopic observation of magnon bound states and their
  dynamics},
\newblock Nature \textbf{502}, 76 (2013),
\newblock \doi{10.1038/nature12541}.

\bibitem{HFSZKDBG14}
S.~Hild, T.~Fukuhara, P.~Schauß, J.~Zeiher, M.~Knap, E.~Demler, I.~Bloch and
  C.~Gross,
\newblock \emph{Far-from-equilibrium spin transport in {H}eisenberg quantum
  magnets},
\newblock Phys. Rev. Lett. \textbf{113}, 147205 (2014),
\newblock \doi{10.1103/PhysRevLett.113.147205}.

\bibitem{Jepsenetal20}
P.~N. Jepsen, J.~Amato-Grill, I.~Dimitrova, W.~W. Ho, E.~Demler and
  W.~Ketterle,
\newblock \emph{Spin transport in a tunable {H}eisenberg model realized with
  ultracold atoms},
\newblock Nature (London) \textbf{588}, 403 (2020),
\newblock \doi{10.1038/s41586-020-3033-y}.

\bibitem{Weietal21pp}
D.~Wei, A.~Rubio-Abadal, B.~Ye, F.~Machado, J.~Kemp, K.~Srakaew, S.~Hollerith,
  J.~Rui, S.~Gopalakrishnan, N.~Y. Yao, I.~Bloch and J.~Zeiher,
\newblock \emph{Quantum gas microscopy of {K}ardar-{P}arisi-{Z}hang
  superdiffusion} (2021), \eprint{2107.00038}.

\bibitem{KlSa02}
A.~Kl\"{u}mper and K.~Sakai,
\newblock \emph{The thermal conductivity of the spin-1/2 {XXZ} chain at
  arbitrary temperature},
\newblock J. Phys. A \textbf{35}, 2173 (2002),
\newblock \doi{10.1088/0305-4470/35/9/307}.

\bibitem{SaKl03}
K.~Sakai and A.~Kl\"umper,
\newblock \emph{Non-dissipative thermal transport in the massive regimes of the
  {XXZ} chain},
\newblock J. Phys. A \textbf{36}, 11617 (2003),
\newblock \doi{10.1088/0305-4470/36/46/006}.

\bibitem{Zotos17}
X.~Zotos,
\newblock \emph{A {TBA} approach to thermal transport in the {XXZ} {H}eisenberg
  model},
\newblock J. Stat. Mech.: Theor. Exp. \textbf{2017}, P103101 (2017),
\newblock \doi{10.1088/1742-5468/aa8c13}.

\bibitem{Suzuki85}
M.~Suzuki,
\newblock \emph{Transfer-matrix method and {Monte Carlo} simulation in quantum
  spin systems},
\newblock Phys. Rev. B \textbf{31}, 2957 (1985),
\newblock \doi{10.1103/PhysRevB.31.2957}.

\bibitem{SAW90}
J.~Suzuki, Y.~Akutsu and M.~Wadati,
\newblock \emph{A new approach to quantum spin chains at finite temperature},
\newblock J. Phys. Soc. Jpn. \textbf{59}, 2667 (1990),
\newblock \doi{10.1143/JPSJ.59.2667}.

\bibitem{Kluemper93}
A.~Kl\"umper,
\newblock \emph{Thermodynamics of the anisotropic spin-1/2 {H}eisenberg chain
  and related quantum chains},
\newblock Z. Phys. B \textbf{91}, 507 (1993),
\newblock \doi{10.1007/BF01316831}.

\bibitem{ShSu90}
B.~S. Shastry and B.~Sutherland,
\newblock \emph{Twisted boundary conditions and effective mass in
  {H}eisenberg-{I}sing and {H}ubbard rings},
\newblock Phys. Rev. Lett. \textbf{65}, 243 (1990),
\newblock \doi{10.1103/PhysRevLett.65.243}.

\bibitem{Mazur69}
P.~Mazur,
\newblock \emph{Non-ergodicity of phase functions in certain systems},
\newblock Physica \textbf{43}, 533 (1969),
\newblock \doi{10.1016/0031-8914(69)90185-2}.

\bibitem{ZNP97}
X.~Zotos, F.~Naef and P.~Prelovsek,
\newblock \emph{Transport and conservation laws},
\newblock Phys. Rev. B \textbf{55}, 11029 (1997),
\newblock \doi{10.1103/PhysRevB.55.11029}.

\bibitem{Prosen11}
T.~Prosen,
\newblock \emph{Open {XXZ} spin chain: Nonequilibrium steady state and a strict
  bound on ballistic transport},
\newblock Phys. Rev. Lett. \textbf{106}, 217206 (2011),
\newblock \doi{10.1103/PhysRevLett.106.217206}.

\bibitem{PereiraPasquier}
R.~G. Pereira, V.~Pasquier, J.~Sirker and I.~Affleck,
\newblock \emph{Exactly conserved quasilocal operators for the {XXZ} spin
  chain},
\newblock J. Stat. Mech p. P09037 (2014),
\newblock \doi{10.1088/1742-5468/2014/09/P09037}.

\bibitem{ProsenIlievski}
T.~Prosen and E.~Ilievski,
\newblock \emph{Families of quasilocal conservation laws and quantum spin
  transport},
\newblock Phys. Rev. Lett. \textbf{111}, 057203 (2013),
\newblock \doi{10.1103/PhysRevLett.111.057203}.

\bibitem{MKP17}
M.~Medenjak, C.~Karrasch and T.~Prosen,
\newblock \emph{Lower bounding diffusion constant by the curvature of {D}rude
  weight},
\newblock Phys. Rev. Lett. \textbf{119}, 080602 (2017),
\newblock \doi{10.1103/PhysRevLett.119.080602}.

\bibitem{IDMP18}
E.~Ilievski, J.~de~Nardis, M.~Medenjak and T.~Prosen,
\newblock \emph{Superdiffusion in one-dimensional quantum lattice models},
\newblock Phys. Rev. Lett. \textbf{121}, 230602 (2018),
\newblock \doi{10.1103/PhysRevLett.121.230602}.

\bibitem{FujimotoKawakami}
S.~Fujimoto and N.~Kawakami,
\newblock \emph{Drude-weight at finite temperatures for some nonintegrable
  quantum systems in one dimension},
\newblock Phys. Rev. Lett. \textbf{90}, 197202 (2003),
\newblock \doi{10.1103/PhysRevLett.90.197202}.

\bibitem{Zotos99}
X.~Zotos,
\newblock \emph{Finite temperature {D}rude weight of the one-dimensional
  spin-1/2 {H}eisenberg model},
\newblock Phys. Rev. Lett. \textbf{82}, 1764 (1999),
\newblock \doi{10.1103/PhysRevLett.82.1764}.

\bibitem{BFKS05}
J.~Benz, T.~Fukui, A.~Kl\"umper and C.~Scheeren,
\newblock \emph{On the finite temperature {D}rude weight of the anisotropic
  {H}eisenberg chain},
\newblock J. Phys. Soc. Jpn. \textbf{74}, 181 (2005),
\newblock \doi{10.1143/JPSJS.74S.181}.

\bibitem{Orbach58}
R.~Orbach,
\newblock \emph{Linear antiferromagnetic chain with anisotropic coupling},
\newblock Phys. Rev. \textbf{112}, 309 (1958),
\newblock \doi{10.1103/PhysRev.112.309}.

\bibitem{BVV83}
O.~Babelon, H.~J. de~Vega and C.~M. Viallet,
\newblock \emph{Analysis of the {B}ethe {A}nsatz equations of the {XXZ} model},
\newblock Nucl. Phys. B \textbf{220}, 13 (1983),
\newblock \doi{10.1016/0550-3213(83)90131-1}.

\bibitem{Gaudin71}
M.~Gaudin,
\newblock \emph{Thermodynamics of the {H}eisenberg-{I}sing ring for {$\Delta >
  1$}},
\newblock Phys. Rev. Lett. \textbf{26}, 1301 (1971),
\newblock \doi{doi.org/10.1103/PhysRevLett.26.1301}.

\bibitem{TaSu72}
M.~Takahashi and M.~Suzuki,
\newblock \emph{One-dimensional anisotropic {H}eisenberg model at finite
  temperatures},
\newblock Prog. Theor. Phys. \textbf{48}, 2187 (1972),
\newblock \doi{10.1143/PTP.48.2187}.

\bibitem{KSS98}
A.~Kuniba, K.~Sakai and J.~Suzuki,
\newblock \emph{Continued fraction {TBA} and functional relations in {XXZ}
  model at root of unity},
\newblock Nucl. Phys. B \textbf{525}, 597 (1998),
\newblock \doi{10.1016/S0550-3213(98)00300-9}.

\bibitem{TSK01}
M.~Takahashi, M.~Shiroishi and A.~Kl\"umper,
\newblock \emph{Equivalence of {TBA} and {QTM}},
\newblock J. Phys. A \textbf{34}, L187 (2001),
\newblock \doi{10.1088/0305-4470/34/13/105}.

\bibitem{ViWo84}
A.~Virosztek and F.~Woynarovich,
\newblock \emph{Degenerated ground states and excited states of the ${S} =
  \frac12$ anisotropic antiferromagnetic {H}eisenberg chain in the easy axis
  region},
\newblock J. Phys. A \textbf{17}, 3029 (1984),
\newblock \doi{10.1088/0305-4470/17/15/020}.

\bibitem{DGKS15a}
M.~Dugave, F.~G\"ohmann, K.~K. Kozlowski and J.~Suzuki,
\newblock \emph{On form factor expansions for the {XXZ} chain in the massive
  regime},
\newblock J. Stat. Mech.: Theor. Exp. p. P05037 (2015),
\newblock \doi{10.1088/1742-5468/2015/05/P05037}.

\bibitem{JiMi95}
M.~Jimbo and T.~Miwa,
\newblock \emph{Algebraic Analysis of Solvable Lattice Models},
\newblock American Mathematical Society (1995).

\bibitem{BCK96}
A.~H. Bougourzi, M.~Couture and M.~Kacir,
\newblock \emph{Exact two-spinon dynamical correlation function of the
  {H}eisenberg model},
\newblock Phys. Rev. B \textbf{54}, R12669 (1996),
\newblock \doi{10.1103/PhysRevB.54.R12669}.

\bibitem{BKM98}
A.~H. Bougourzi, M.~Karbach and G.~M\"uller,
\newblock \emph{Exact two-spinon dynamic structure factor of the
  one-dimensional $s = 1/2$ {H}eisenberg-{I}sing antiferromagnet},
\newblock Phys. Rev. B \textbf{57}, 11429 (1998),
\newblock \doi{10.1103/PhysRevB.57.11429}.

\bibitem{CaHa06}
J.-S. Caux and R.~Hagemans,
\newblock \emph{The 4-spinon dynamical structure factor of the {H}eisenberg
  chain},
\newblock J. Stat. Mech.: Theor. Exp. p. P12013 (2006),
\newblock \doi{10.1088/1742-5468/2006/12/P12013}.

\bibitem{CKSW12}
J.-S. Caux, H.~Konno, M.~Sorrell and R.~Weston,
\newblock \emph{Exact form-factor results for the longitudinal structure factor
  of the massless {XXZ} model in zero field},
\newblock J. Stat. Mech.: Theor. Exp. p. P01007 (2012),
\newblock \doi{10.1088/1742-5468/2012/01/P01007}.

\bibitem{CaMa05}
J.-S. Caux and J.~M. Maillet,
\newblock \emph{Computation of dynamical correlation functions of {H}eisenberg
  chains in a field},
\newblock Phys. Rev. Lett. \textbf{95}, 077201 (2005),
\newblock \doi{10.1103/PhysRevLett.95.077201}.

\bibitem{CMP08}
J.-S. Caux, J.~Mossel and I.~P. Castillo,
\newblock \emph{The two-spinon transverse structure factor of the gapped
  {H}eisenberg antiferromagnetic chain},
\newblock J. Stat. Mech.: Theor. Exp. p. P08006 (2008),
\newblock \doi{10.1088/1742-5468/2008/08/P08006}.

\bibitem{KKMST12}
N.~Kitanine, K.~K. Kozlowski, J.~M. Maillet, N.~A. Slavnov and V.~Terras,
\newblock \emph{Form factor approach to dynamical correlation functions in
  critical models},
\newblock J. Stat. Mech.: Theor. Exp. p. P09001 (2012),
\newblock \doi{10.1088/1742-5468/2012/09/P09001}.

\bibitem{Kozlowski18}
K.~K. Kozlowski,
\newblock \emph{On the thermodynamic limit of form factor expansions of
  dynamical correlation functions in the massless regime of the {XXZ} spin 1/2
  chain},
\newblock J. Math. Phys. \textbf{59}, 091408 (2018),
\newblock \doi{10.1063/1.5021892}.

\bibitem{Kozlowski21}
K.~K. Kozlowski,
\newblock \emph{On singularities of dynamic response functions in the massless
  regime of the {XXZ} spin-1/2 chain},
\newblock J. Math. Phys. \textbf{62}, 063507 (93pp) (2021),
\newblock \doi{10.1063/5.003651},
\newblock \eprint{math-ph 1811.06076}.

\bibitem{GKKKS17}
F.~G\"ohmann, M.~Karbach, A.~Kl\"umper, K.~K. Kozlowski and J.~Suzuki,
\newblock \emph{Thermal form-factor approach to dynamical correlation functions
  of integrable lattice models},
\newblock J. Stat. Mech.: Theor. Exp. p. 113106 (2017),
\newblock \doi{10.1088/1742-5468/aa9678}.

\bibitem{FaEs13}
M.~Fagotti and F.~H.~L. Essler,
\newblock \emph{Stationary behaviour of observables after a quantum quench in
  the spin-1/2 {H}eisenberg {XXZ} chain},
\newblock J. Stat. Mech.: Theor. Exp. p. P07012 (2013),
\newblock \doi{10.1088/1742-5468/2013/07/p07012}.

\bibitem{GKS20a}
F.~G\"ohmann, K.~K. Kozlowski and J.~Suzuki,
\newblock \emph{High-temperature analysis of the transverse dynamical two-point
  correlation function of the {XX} quantum-spin chain},
\newblock J. Math. Phys. \textbf{61}, 013301 (2020),
\newblock \doi{10.1063/1.5111039}.

\bibitem{GKS20b}
F.~G\"ohmann, K.~K. Kozlowski and J.~Suzuki,
\newblock \emph{Long-time large-distance asymptotics of the transversal
  correlation functions of the {XX} chain in the space-like regime},
\newblock Lett. Math. Phys. \textbf{110}, 1783 (2020),
\newblock \doi{10.1007/s11005-020-01276-y}.

\bibitem{CIKT92}
F.~Colomo, A.~G. Izergin, V.~E. Korepin and V.~Tognetti,
\newblock \emph{Correlators in the {H}eisenberg {XXO} chain as {F}redholm
  determinants},
\newblock Phys. Lett. A \textbf{169}, 243 (1992),
\newblock \doi{10.1016/0375-9601(92)90452-R}.

\bibitem{IIKS93b}
A.~R. Its, A.~G. Izergin, V.~E. Korepin and N.~Slavnov,
\newblock \emph{Temperature correlations of quantum spins},
\newblock Phys. Rev. Lett. \textbf{70}, 1704 (1993),
\newblock \doi{10.1103/PhysRevLett.70.1704}.

\bibitem{DGKS15b}
M.~Dugave, F.~G\"ohmann, K.~K. Kozlowski and J.~Suzuki,
\newblock \emph{Low-temperature spectrum of correlation lengths of the {XXZ}
  chain in the antiferromagnetic massive regime},
\newblock J. Phys. A \textbf{48}, 334001 (2015),
\newblock \doi{10.1088/1751-8113/48/33/334001}.

\bibitem{BGKS21a}
C.~Babenko, F.~G\"ohmann, K.~K. Kozlowski and J.~Suzuki,
\newblock \emph{A thermal form factor series for the longitudinal two-point
  function of the {H}eisenberg-{I}sing chain in the antiferromagnetic massive
  regime},
\newblock J. Math. Phys. \textbf{62}, 041901 (2021),
\newblock \doi{10.1063/5.0039863}.

\bibitem{BGKSS21}
C.~Babenko, F.~G\"ohmann, K.~K. Kozlowski, J.~Sirker and J.~Suzuki,
\newblock \emph{Exact real-time longitudinal correlation functions of the
  massive {XXZ} chain},
\newblock Phys. Rev. Lett. \textbf{126}, 210602 (2021),
\newblock \doi{10.1103/PhysRevLett.126.210602}.

\bibitem{DGKS16b}
M.~Dugave, F.~G\"ohmann, K.~K. Kozlowski and J.~Suzuki,
\newblock \emph{Thermal form factor approach to the ground-state correlation
  functions of the {XXZ} chain in the antiferromagnetic massive regime},
\newblock J. Phys. A \textbf{49}, 394001 (2016),
\newblock \doi{10.1088/1751-8113/49/39/394001}.

\bibitem{BGKS06}
H.~Boos, F.~G\"ohmann, A.~Kl\"umper and J.~Suzuki,
\newblock \emph{Factorization of multiple integrals representing the density
  matrix of a finite segment of the {H}eisenberg spin chain},
\newblock J. Stat. Mech.: Theor. Exp. p. P04001 (2006),
\newblock \doi{10.1088/1742-5468/2006/04/P04001}.

\bibitem{BGKS07}
H.~Boos, F.~G\"ohmann, A.~Kl\"umper and J.~Suzuki,
\newblock \emph{Factorization of the finite temperature correlation functions
  of the {XXZ} chain in a magnetic field},
\newblock J. Phys. A \textbf{40}, 10699 (2007),
\newblock \doi{10.1088/1751-8113/40/35/001}.

\bibitem{BJMST08a}
H.~Boos, M.~Jimbo, T.~Miwa, F.~Smirnov and Y.~Takeyama,
\newblock \emph{Hidden {G}rassmann structure in the {XXZ} model {II}: creation
  operators},
\newblock Comm. Math. Phys. \textbf{286}, 875 (2009),
\newblock \doi{10.1007/s00220-008-0617-z}.

\bibitem{JMS08}
M.~Jimbo, T.~Miwa and F.~Smirnov,
\newblock \emph{Hidden {G}rassmann structure in the {XXZ} model {III}:
  introducing {M}atsubara direction},
\newblock J. Phys. A \textbf{42}, 304018 (2009),
\newblock \doi{10.1088/1751-8113/42/30/304018}.

\bibitem{BoGo09}
H.~Boos and F.~G\"ohmann,
\newblock \emph{On the physical part of the factorized correlation functions of
  the {XXZ} chain},
\newblock J. Phys. A \textbf{42}, 315001 (2009),
\newblock \doi{10.1088/1751-8113/42/31/315001}.

\bibitem{WhWa63ch21}
E.~T. Whittaker and G.~N. Watson,
\newblock \emph{A {C}ourse of {M}odern {A}nalysis}, chap.~21,
\newblock Cambridge University Press, fourth edn.,
\newblock \doi{10.1017/CBO9780511608759} (1963).

\bibitem{EnSi12}
T.~Enss and J.~Sirker,
\newblock \emph{Lightcone renormalization and quantum quenches in
  one-dimensional {H}ubbard models},
\newblock New J. Phys. \textbf{14}, 023008 (2012),
\newblock \doi{10.1088/1367-2630/14/2/023008}.

\bibitem{EAS17}
T.~Enss, F.~Andraschko and J.~Sirker,
\newblock \emph{Many-body localization in infinite chains},
\newblock Phys. Rev. B \textbf{95}, 045121 (2017),
\newblock \doi{10.1103/PhysRevB.95.045121}.

\bibitem{DGKS16a}
M.~Dugave, F.~G\"ohmann, K.~K. Kozlowski and J.~Suzuki,
\newblock \emph{Asymptotics of correlation functions of the
  {H}eisenberg-{I}sing chain in the easy-axis regime},
\newblock J. Phys. A \textbf{49}, 07LT01 (2016),
\newblock \doi{10.1088/1751-8113/49/7/07LT01}.

\bibitem{BGKKW12}
M.~Brockmann, F.~G\"ohmann, M.~Karbach, A.~Kl\"umper and A.~Wei{\ss}e,
\newblock \emph{Absorption of microwaves by the one-dimensional spin-1/2
  {H}eisenberg-{I}sing magnet},
\newblock Phys. Rev. B \textbf{85}, 134438 (2012),
\newblock \doi{10.1103/PhysRevB.85.134438}.

\bibitem{BariAdler}
R.~A. Bari, D.~Adler and R.~V. Lange,
\newblock \emph{Electrical conductivity in narrow energy bands},
\newblock Phys. Rev. B \textbf{2}, 2898 (1970),
\newblock \doi{10.1103/PhysRevB.2.2898}.

\bibitem{KKM14}
C.~Karrasch, D.~M. Kennes and J.~E. Moore,
\newblock \emph{Transport properties of the one-dimensional {H}ubbard model at
  finite temperature},
\newblock Phys. Rev. B \textbf{90}, 155104 (2014),
\newblock \doi{10.1103/PhysRevB.90.155104}.

\bibitem{GGKS20}
F.~G\"ohmann, S.~Goomanee, K.~K. Kozlowski and J.~Suzuki,
\newblock \emph{Thermodynamics of the spin-1/2 {H}eisenberg-{I}sing chain at
  high temperatures: a rigorous approach},
\newblock Comm. Math. Phys. \textbf{377}, 623 (2020),
\newblock \doi{10.1007/s00220-020-03749-6}.

\bibitem{Kozlowski20app}
K.~K. Kozlowski,
\newblock \emph{On convergence of form factor expansions in the infinite volume
  quantum {S}inh-{G}ordon model in 1+1 dimensions},
\newblock preprint, arXiv:2007.01740,
\newblock \doi{10.48550/arXiv.2007.01740} (2020).

\bibitem{GaRa04}
G.~Gasper and M.~Rahman,
\newblock \emph{Basic Hypergeometric Series}, vol.~96 of \emph{Encyclopedia of
  Mathematics and its Applications},
\newblock Cambridge University Press, scd. edn.,
\newblock \doi{10.1017/CBO9780511526251} (2004).

\bibitem{Thebook}
F.~H.~L. Essler, H.~Frahm, F.~G\"ohmann, A.~Kl\"umper and V.~E. Korepin,
\newblock \emph{The {O}ne-{D}imensional {H}ubbard {M}odel},
\newblock Cambridge University Press,
\newblock \doi{10.1017/CBO9780511534843} (2005).

\bibitem{Goehmann21}
F.~G\"ohmann,
\newblock \emph{Introduction to solid state physics},
\newblock Lecture Notes (2021), \eprint{arXiv:2101.01780}.

\bibitem{GiamarchiShastri}
T.~Giamarchi and B.~S. Shastry,
\newblock \emph{Persistent currents in a one-dimensional ring for a disordered
  {H}ubbard model},
\newblock Phys. Rev. B \textbf{51}, 10915 (1995),
\newblock \doi{10.1103/PhysRevB.51.10915}.

\end{thebibliography}

\end{document}